\pdfoutput=1
\documentclass{article}
\usepackage{iclr2027_conference,times}

\usepackage[utf8]{inputenc}
\usepackage[T1]{fontenc}
\usepackage{amsmath,amssymb,amsthm}
\usepackage{enumitem}
\usepackage[lambda,adversary,advantage,operators,sets,landau,probability,notions,logic,ff,mm,primitives,events,complexity,oracles,asymptotics,keys]{cryptocode}
\usepackage{todonotes}
\usepackage{booktabs}
\usepackage{multirow}
\usepackage[hidelinks]{hyperref}
\usepackage{url}
\usepackage{xcolor}
\usepackage{graphicx}
\usepackage{tikz}
\usetikzlibrary{positioning, arrows.meta, shapes.geometric}
\usepackage{stmaryrd} 
\usepackage{cleveref}

\iclrfinalcopy 
\newcommand{\pdrartifacturl}{https://github.com/tremblaythibaultl/pdr-artifact}
\newcommand{\vpirurl}{https://github.com/tremblaythibaultl/vpir}

\newtheorem{definition}{Definition}
\newtheorem{theorem}{Theorem}

\newcommand{\R}{\mathcal{R}}
\newcommand{\mat}[1]{\mathbf{#1}}
\newcommand{\vect}[1]{\mathbf{#1}}

\newcommand{\calE}{\mathcal{E}}

\newcommand{\calU}{\mathcal{U}}

\newcommand{\innerp}[2]{\left\langle #1, #2\right\rangle}

\newcommand{\acc}{\mathsf{Accept}}

\newcommand{\inftynorm}[1]{\| #1 \|_{\infty}}

\newcommand{\cmark}{\ensuremath{\checkmark}}
\newcommand{\xmark}{\ensuremath{\times}}

\title{Quantization Enables Private Dense Retrieval against Malicious Service Providers}
\author{\mdseries
\begin{tabular}[t]{@{}l@{\hspace{3em}}l@{}}
\bf Louis Tremblay Thibault & \bf Sofiane Azogagh \\
\'Ecole de technologie sup\'erieure and Mila & Eurecom \\[1.5ex]
\bf Marc-Olivier Killijian & \bf Ulrich A\"ivodji \\
Universit\'e du Qu\'ebec \`a Montr\'eal & \'Ecole de technologie sup\'erieure and Mila
\end{tabular}}

\begin{document}

\maketitle
\lhead{Preprint.} 

\begin{abstract}
Dense retrieval, the key component of Retrieval Augmented Generation (RAG), retrieves the most relevant documents by comparing dense vector representations of queries and passages from a large corpus. In privacy-sensitive applications, the server observes the query and controls which evidence is returned, creating both confidentiality and integrity risks. We formulate private dense retrieval as providing query privacy and retrieval integrity against a malicious server, and develop a two-round cryptographic protocol that provides both guarantees. Our protocol reduces private and verifiable retrieval to multiplication of a committed matrix by an encrypted vector and uses low-bit quantization to make this computation practical. We evaluate the resulting trade-off between cryptographic cost, retrieval quality, and downstream RAG accuracy across six embedding models, four language models, and corpora of up to 2.68 million passages. Our results show that, with a clipped quantizer, three-bit quantization largely preserves retrieval quality and downstream accuracy, while a private query over a corpus the size of a clinical reference requires one to three minutes of server time. These results suggest that private dense retrieval is already practical for moderately sized, privacy-sensitive corpora when minute-scale latency is acceptable.
\end{abstract}

\section{Introduction}
\label{sec:intro}

Dense retrieval is a central paradigm for large-scale information access.
Bi-encoder models~\citep{karpukhin2020dense,reimers2019sentence,ni2022large} represent documents and queries as vectors in a shared embedding space and answer queries by nearest-neighbor search.
This is the mechanism underlying semantic search, open-domain question answering and the retrieval stage of retrieval-augmented generation (RAG)~\citep{lewis2020retrieval}.
For exact retrieval by inner-product scoring, the interaction can be expressed as a matrix-vector product:
the server holds an embedding matrix $\mat{A} \in \mathbb{R}^{m \times n}$, where each row represents a document, and the user submits a query embedding $\vect{x} \in \mathbb{R}^n$. In such a setting, retrieval reduces to computing the score vector $\mat{A}\vect{x}$ and fetching the highest-scoring documents.

However, this interaction exposes query embeddings to the server in the clear, potentially revealing sensitive information about the user's intent. In particular, these embeddings can be linked across queries and inverted to recover information about the underlying text~\citep{morris2023text,song2020leakage}.
Consider clinical decision support, where AI-assisted synthesis of the medical literature is increasingly common\footnote{$39\%$ of US physicians report incorporating summaries of medical research and standards of care into practice, up from $13\%$ two years earlier~\citep{ama2026ai}.} and physicians submit queries to tools such as \emph{UpToDate} and \emph{OpenEvidence}, which report reaching over three million clinicians worldwide and a majority of US physicians, respectively~\citep{uptodate2026,openevidence2026}.
Every such query is visible to the provider, to a compromised insider or to anyone who breaches the service or its logs.\footnote{Medical IT providers are attractive targets for cyberattacks in practice, as in the 2026 Cegedim Sant\'e incident that exposed 15~million medical files~\citep{cegedim}.}
By cross-referencing with a clinician's appointment schedule, an adversary could link individual patients to the topics their physician searched and reconstruct a working diagnosis.
Such linkage can disclose protected health information subject to the HIPAA Privacy Rule~\citep{HIPAA}.

Beyond query privacy, the client also needs guarantees on the integrity of the retrieved evidence. 
Indeed, the server currently decides what evidence reaches the user or the downstream language model.
One that silently rescores documents or substitutes the returned passage can steer diagnoses or generated answers, and an encrypted-but-unverified protocol gives the clinician no way to notice.
Existing systems provide one of these two guarantees but not both, and those that encrypt the query assume an ``honest-but-curious'' server, which contradicts the threat model that motivates encryption in the first place (\S\ref{sec:related}).
Retrieval integrity also bounds \emph{indirect prompt injection}~\citep{Greshake2023NotWY} through the retrieval channel: once a corpus has been audited and committed, neither the provider nor an attacker with write access to the corpus can slip adversarial instructions into the passages that reach the reader, since the client accept only content consistent with the audited digest (\S\ref{sec:protocol}).

\begin{figure}[t]
\centering
\includegraphics[width=0.9\textwidth,trim={6 20 12 10},clip]{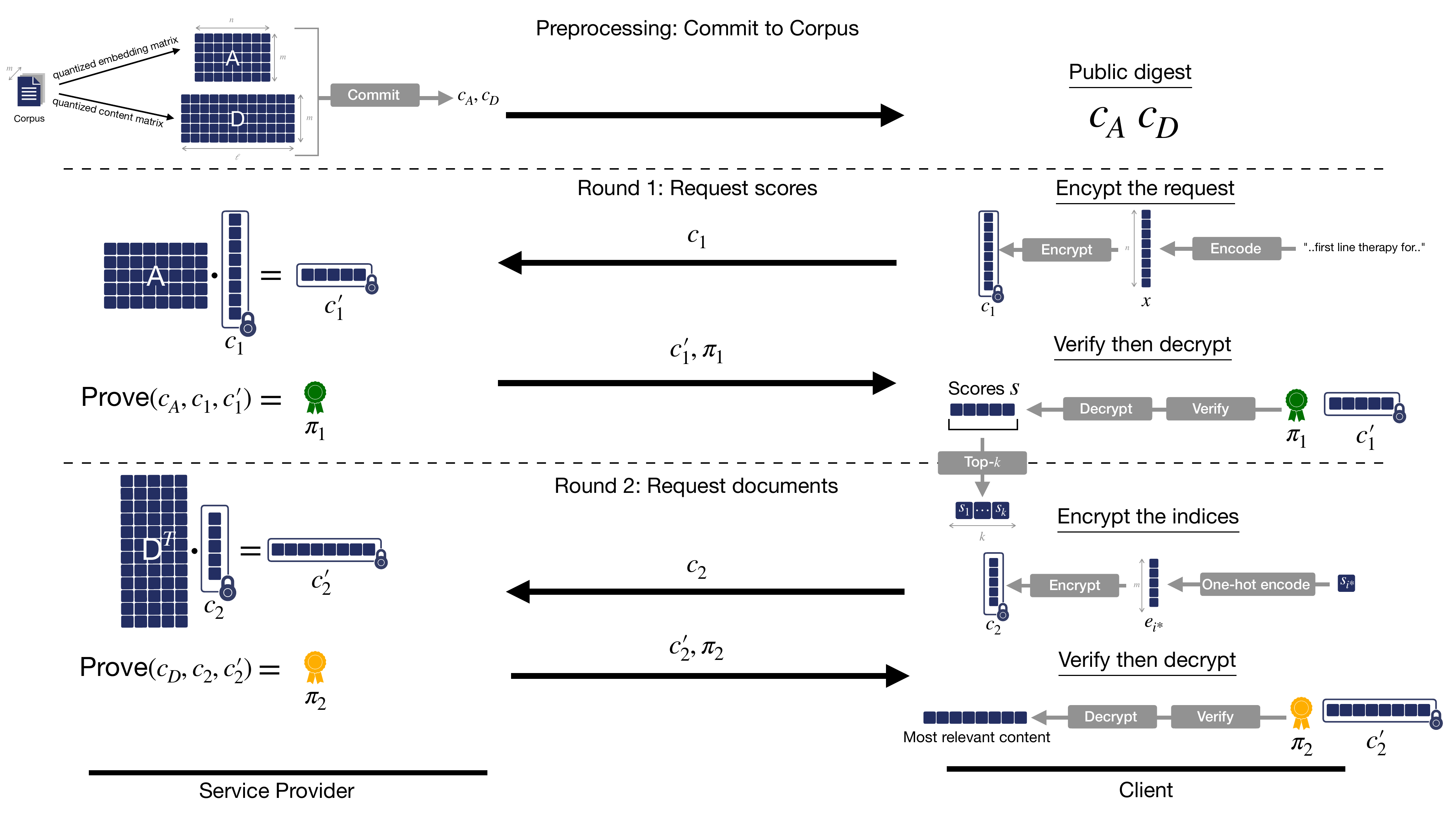}
\caption{Private dense retrieval. Round~1 returns encrypted scores, from which the client selects $i^*$; Round~2 fetches $\mat{D}_{i^*}$. Each proof is verified against the committed $\mat{A}$ or $\mat{D}$ before decryption.}
\label{fig:overview}
\end{figure}

\paragraph{Background.}
Homomorpic encryption (HE)~\citep{10.1145/1536414.1536440} is a form of encryption that allows computations on encrypted data without decrypting it first.
The party holding the encryption key can decrypt the result of the computation, but the party performing the computation learns nothing about the underlying data.
Over nearly two decades, HE has evolved from a theoretical topic to a vast field of privacy-preserving applications.
For example, HE has been used to enable Private Information Retrieval (PIR)~\citep{chor1998pir}, where a client can retrieve an item from a server's database without revealing to the server which item it is has retrieved.
It is crucial to understand that homomorphically encrypted ciphertexts must contain \emph{noise} for security.
If not carefully managed, noise accumulation over homomorphic operations can lead to incorrect computation results.
As such, smaller noise allows for more operations over the ciphertexts, but guarantees a lower security level.
This trade-off is central to our work and we perform in \S\ref{subsec:correctness} the careful balancing act of choosing cryptographic parameters that provide adequate levels of both security and correctness for the tasks at hand.

\paragraph{Our approach.}
Our starting point is the observation that the two stages of dense retrieval, scoring and document fetch, reduce to a product between a public matrix held by the server (\emph{e.g., }provider) and a private, encrypted vector held by the client (\emph{e.g.,} clinician).
We instantiate both with a single maliciously secure primitive, verifiable matrix-vector multiplication under HE~\citep{cryptoeprint:2026/027}.
The server cryptographically commits to a matrix, multiplies its committed matrix by the client's encrypted vector and returns the encrypted result with a succinct proof, which the client verifies before decrypting.
The primitive guarantees integrity by construction, and we show in \S\ref{sec:protocol} that verifying before decrypting also closes the reaction-attack channel~\citep{cryptoeprint:2016/1164} that honest-but-curious designs leave open.
The resulting two-round protocol, illustrated in Figure~\ref{fig:overview}, is exact (every document is scored, with no clustering or hashing approximation), secure against a malicious server (any manipulation of scores or of the returned document is detected except with negligible probability), client-lightweight (per round the client encrypts one vector, verifies one proof and decrypts one result) and modular (both rounds use one primitive behind one interface, so any faster instantiation may be swapped in directly).

\paragraph{Setting.}
The naive way to hide a query is to send the corpus to the client and have it search locally.
That option is not available here: the clinical references that physicians consult are proprietary and licensed per seat, so the provider grants query-time access rather than a copy.
As such, the setting we consider is a server-held corpus queried by a client that must not reveal its queries.
We target deferred-response settings, such as literature requests to medical librarians or tumor-board preparation, where users expect a researched answer rather than an immediate response and minute-scale latency is acceptable; we evaluate the protocol's cost in this regime (\S\ref{subsec:costs}).

\paragraph{Contributions.}
\begin{enumerate}[noitemsep,topsep=2pt]
    \item We formalize private dense retrieval as requiring both query privacy and retrieval integrity against a malicious server, and reduce it to two instances of verifiable matrix-vector multiplication~(\S\ref{sec:problem}). 
    We give a two-round protocol realizing this functionality, with exact retrieval and a verify-before-decrypt mechanism that provably blocks reaction attacks~(\S\ref{sec:protocol}).
    \item Using the protocol as an instrument, we experimentally locate the operating point of private dense retrieval at $b{=}3$ bits of quantization: retrieval quality and downstream RAG accuracy hold there for six encoders and four readers, on corpora of up to $2.68$M passages~(\S\ref{subsec:retrieval}, \S\ref{subsec:downstream}).
    \item We measure the cost of a private, verifiable query in time and dollars at corpus sizes up to $10^6$ chunks, and show that bit-width buys performance, with a $\approx 2 \times$ increase between $b{=}2$ and $b{=}3$~(\S\ref{subsec:correctness}, \S\ref{subsec:costs}).
    At $b{=}3$ a corpus the size of a clinical reference is served in one to three minutes, while the literature behind it takes hours to serve.
\end{enumerate}
\section{Private Dense Retrieval}
\label{sec:problem}

The setting comprises two interactive parties.
The server holds a corpus of $m$ documents, represented by an embedding matrix $\mat{A} \in \ZZ_p^{m \times n}$ (row $i$ is the quantized embedding of document $i$) and a content matrix $\mat{D} \in \ZZ_p^{m \times \ell}$ (row $i$ is the content of document $i$, padded to length $\ell$).
The client holds a private query embedding $\vect{x} \in \ZZ_p^n$.
Embeddings are quantized from $\mathbb{R}$ to $\ZZ_p$ by scalar quantization (Appendix~\ref{sec:quantization}). This is further discussed in \S\ref{sec:experiments}.
This abstracts a common deployment pattern in RAG and enterprise search: the client trusts the embedding model that produced $\vect{x}$, but does not trust the retrieval service with the confidentiality of the query or the integrity of the returned evidence.

\begin{definition}[Private Dense Retrieval]
\label{def:pdr}
A \emph{private dense retrieval protocol} allows the client to
(1)~obtain the similarity scores $\vect{s} = \mat{A} \cdot \vect{x} \in \ZZ_p^m$, and
(2)~retrieve the document $\mat{D}_{i^*}$ for $i^* = \arg\max_i s_i$,
subject to:
\begin{enumerate}[noitemsep,topsep=1pt]
  \item \textbf{privacy}: the server learns no information about $\vect{x}$ or $i^*$ beyond public parameters; and
  \item \textbf{integrity}: the client accepts only scores and documents consistent with the server's committed database, except with negligible probability.
\end{enumerate}
\end{definition}

The server is actively malicious: it may deviate arbitrarily from the protocol to learn the query or to alter the result.
As in standard PIR, we do not require database privacy against the client; our confidentiality goal concerns the client's query.
The definition extends to top-$k$ retrieval (\S\ref{sec:topk}).

\paragraph{One primitive suffices.}
\label{sec:primitive}
Both requirements of Definition~\ref{def:pdr} reduce to verifiable multiplication of a public plaintext matrix by an encrypted vector.
The primitive provides privacy (ciphertexts reveal nothing about the encrypted vector by semantic security under the GLWE assumption), soundness (no server strategy makes the client accept anything but the product with the committed matrix, except with negligible probability) and succinctness (digest and proof are small, and verifying is far cheaper than computing the product).
Formal definitions are in Appendix~\ref{sec:secure-matvec}; \citet{cryptoeprint:2026/027} give a concrete instantiation, which our protocol treats as a module. 
Improvements to the primitive thus transfer directly to our protocol.

\paragraph{Protocol.}
\label{sec:protocol}

The protocol runs two rounds of the same sub-protocol (Figure~\ref{fig:overview}).
Once per database, the server commits to $\mat{A}$ and $\mat{D}$ and publishes the digests, which remain fixed for the database's lifetime.
Updates require recomputing only the affected commitment, a cost amortized over many queries.
In Round~1, the client sends $\vect{c}_1 \gets \enc(\sk_1, \vect{x})$ and the server returns $\vect{c}'_1 \gets \mat{A} \cdot \vect{c}_1$ with a proof $\pi_1$ against the committed $\mat{A}$.
From this the client decrypts the exact score vector and computes $i^* \gets \arg\max_i s_i$ locally.
In Round~2, the client sends $\vect{c}_2 \gets \enc(\sk_2, \vect{e}_{i^*})$, for $\vect{e}_{i^*} \in \{0,1\}^m$ the the one-hot encoding of $i^*$.
The server returns $\mat{D}^\top \cdot \vect{c}_2$ with a proof $\pi_2$ against the committed $\mat{D}$, from which the client decrypts $\mat{D}_{i^*}$.
Each round uses a fresh key and the client aborts unless verification succeeds.
The rounds are linked only through $i^*$, which is computed locally. 
The server sees two semantically secure ciphertext vectors under independent keys and learns no information about $\vect{x}$ or $i^*$ beyond public parameters.
The client's Round-1 work, namely downloading $m$ ciphertext slots, decrypting them and scanning for the maximum, is linear in $m$.
We discuss scaling in \S\ref{sec:discussion}.

\paragraph{Verification before decryption.}
That the client verifies before decrypting is essential to security.
If it decrypted an unverified response, its observable behavior (abort, retry, or which document it fetches in Round~2) would become a function of the plaintext.
This is the side channel that reaction attacks exploit~\citep{cryptoeprint:2016/1164}.
Verifying first makes that behavior depend only on the proof's validity, which the server already knows.

\begin{theorem}[Informal; proof sketch in Appendix~\ref{sec:proof}]
\label{thm:security}
If the primitive satisfies the properties of \S\ref{sec:primitive}, the protocol realizes Definition~\ref{def:pdr}: an actively malicious server learns  no information about $\vect{x}$ or $i^*$ beyond public parameters, and cannot cause the client to accept scores or a document inconsistent with the committed database, except with negligible probability.
\end{theorem}

\paragraph{Top-$k$ retrieval.}
\label{sec:topk}
RAG pipelines typically consume $k > 1$ passages.
After Round~1 the client knows the full score vector, so it can select any $k$ indices and fetch each with an independent Round-2 query, multiplying the second round's workload by $k$.
Batching the $k$ queries into one server request multiplies communication by $k$ but barely increases latency, since the online server work of~\citet{cryptoeprint:2026/027} is dominated by a step independent of the number of batched queries (Appendix~\ref{sec:topk-batch}).

\section{Related Work}
\label{sec:related}

\begin{table}[b]
\centering
\caption{Private and verifiable retrieval systems. \emph{Exactness}: scores are exact inner products rather than approximate search. \emph{Malicious security}: the stated privacy and/or integrity guarantees hold against an actively deviating server.}
\label{tab:comparison}
\begin{tabular}{@{}lcccc@{}}
\toprule
 & Privacy & Integrity & Exactness & Malicious security \\
\midrule
Tiptoe~\citep{10.1145/3600006.3613134} & \cmark & \xmark & \xmark & \xmark \\
PIR-RAG~\citep{wang2025pirragprivateinformationretrieval} & \cmark & \xmark & \xmark & \xmark \\
PPMI~\citep{bae2025ppmiprivacypreservingllminteraction} & \cmark & \xmark & \cmark & \xmark \\
Compass~\citep{zhu2024compass} & \cmark & \xmark & \xmark & \xmark \\
VeriRAG~\citep{cryptoeprint:2026/637} & \xmark & \cmark & \xmark & \cmark \\
zkRAG~\citep{cryptoeprint:2026/709} & \xmark & \cmark & \xmark & \cmark \\
\textbf{Ours} & \cmark & \cmark & \cmark & \cmark \\
\bottomrule
\end{tabular}
\end{table}

\paragraph{Encrypted and verifiable retrieval.}
Table~\ref{tab:comparison} situates our protocol relative to prior work.
We take dense retrieval by maximum inner product search~\citep{karpukhin2020dense,reimers2019sentence,ni2022large} as given and make the query private and the result verifiable.
The natural tool for hiding the query is HE, as PPMI~\citep{bae2025ppmiprivacypreservingllminteraction} does for private LLM interaction.
Tiptoe~\citep{10.1145/3600006.3613134} and PIR-RAG~\citep{wang2025pirragprivateinformationretrieval} instead search published cluster centroids locally and fetch documents by private information retrieval (PIR)~\citep{chor1998pir}, for which practical single-server lattice-based constructions include XPIR~\citep{POPETS:ABFK16}. Other systems perform the search itself privately: SANNS~\citep{chen2020sanns}, \citet{servanschreiber2022pann}, Compass~\citep{zhu2024compass} and BiSON~\citep{das2026bison} run approximate $k$-NN under two-party computation or oblivious RAM, at second-to-millisecond latencies on million- to billion-scale corpora.
These systems assume an honest-but-curious server, or two non-colluding servers, and do not provide integrity against an actively deviating server. In interactive encrypted protocols, such deviations can additionally enable reaction attacks that compromise privacy~\citep{cryptoeprint:2016/1164}.
Conversely, VeriRAG~\citep{cryptoeprint:2026/637} and zkRAG~\citep{cryptoeprint:2026/709} prove in zero knowledge that the server ran an approximate search over a committed database, but send the query in the clear.
Our work instead follows the verifiable-computation-over-ciphertext paradigm~\citep{bois2021verifiable,ganesh2023rinocchio,chatel2024veritas,cryptoeprint:2026/027}, enabling query privacy and retrieval integrity while preserving exactness.

\paragraph{Low-bit embeddings and the retrieval-metric/downstream gap.}
Several methods have been proposed in recent years for compressing embeddings for more efficient downstream tasks, including binary hashing~\citep{yamada2021bpr}, jointly learned product quantization~\citep{zhan2021jpq}, Matryoshka representation learning~\citep{kusupati2022matryoshka}, int8-trained commercial encoders~\citep{cohere2023int8} and quantization-aware open ones~\citep{yu2024arctic}.
Clipping the quantization range to a percentile of the values, rather than to their maximum, is standard in post-training quantization~\citep{banner2019posttraining}.
That literature primarily optimizes memory footprint and search efficiency, with representation size scaling linearly with bit-width.
In our encrypted setting, every bit consumes a fixed share of a noise budget that only larger cryptographic parameters can compensate for (\S\ref{subsec:correctness}), which imposes a different cost model.
Separately, ranking metrics are known to predict end-task RAG accuracy poorly, as retrieval and generation quality decouple~\citep{salemi2024erag,cuconasu2024noise}.
Section~\ref{subsec:downstream} studies this decoupling as a function of bit-width, which also sets cryptographic cost, allowing us to select parameters based on downstream rather than retrieval quality alone.

\section{Experimental Analysis}
\label{sec:experiments}

The protocol preserves exact dense retrieval over quantized embeddings: the only approximation is the $b$-bit quantization itself, while all subsequent arithmetic is computed exactly in $\ZZ_p$, with $p \geq 2NB^2$.
In fact, the GLWE plaintext modulus $p$ (cf. \S\ref{sec:glwe}) must be large enough so that $\ZZ_p$ contains the quantized embeddings and their inner products, without overflow.
A larger $b$ requires a larger $p$, which drives up cryptographic parameters.
Holding the retrieval task and security target fixed, the bit-width $b$ is therefore the key parameter coupling retrieval quality and cryptographic cost.
Our experiments therefore ask how far $b$ can be reduced before the cryptographic savings cease to justify the loss in retrieval and downstream quality. We measure this trade-off holistically: \S\ref{subsec:retrieval} measures retrieval quality across bit-widths and corpus sizes, \S\ref{subsec:correctness} gives the conversion from bit-width to cryptographic cost, \S\ref{subsec:costs} the concrete costs at scale and \S\ref{subsec:downstream} the downstream performance.

\paragraph{Corpora.}
We evaluate on two biomedical retrieval corpora, chosen to match the setting of \S\ref{sec:intro}: \emph{SciFact}~\citep{wadden2020fact}, scientific claim verification over biomedical abstracts ($m{=}5{,}183$ documents, $300$ test claims), and \emph{NFCorpus}~\citep{thakur2021beir}, medical retrieval pairing nutrition and health queries with linked PubMed documents ($m{=}3{,}633$, $323$ queries).
Both are BEIR~\citep{thakur2021beir} tasks, so floating-point baselines are comparable with published numbers.
The retrieval unit of the protocol is a chunk of at most $1{,}536$ bytes: $512$ matrix cells each containing $3$ bytes (cf. Table~\ref{tab:benchmarks}).
This allows more than half of documents to fit into a single chunk, all the while keeping the cost of the second round low.
We cut at sentence boundaries with the title repeated in every chunk. 
This accommodates the clinical references of \S\ref{sec:intro}, which are typically several pages long.
The downstream evaluation of \S\ref{subsec:downstream} adds \emph{BioASQ}~\citep{krithara2023bioasq} for biomedical factoid question answering, and the scale study of \S\ref{subsec:costs} uses \emph{Natural Questions}~\citep{kwiatkowski2019natural}, which also serves \S\ref{subsec:downstream} as an out-of-domain control.
BioASQ and Natural Questions are passage collections already fitting within one chunk.
Retrieval quality in \S\ref{subsec:retrieval} is measured on the chunked corpora, ranked per chunk and scored per document, so that the relevance labels and the floating-point baselines remain the published document-level ones; Appendix~\ref{sec:encoder-suite} gives the same table on whole abstracts.

\subsection{Retrieval quality under quantization}
\label{subsec:retrieval}

We evaluate six embedding models of similar performance on both corpora, so that training for low-bit readout is the variable.
Two were trained for low-bit readout: Cohere \texttt{embed-english-v3.0}, trained for int8 and binary output~\citep{cohere2023int8}, and Snowflake \texttt{arctic-embed-l-v2.0}, whose training was quantization-aware~\citep{yu2024arctic,arctic2024card}.
Four were not: \texttt{bge-large-en-v1.5}~\citep{xiao2023bge}, \texttt{e5-large-v2}~\citep{wang2022e5}, \texttt{gte-large}~\citep{li2023gte} and \texttt{mxbai-embed-large-v1}~\citep{li2023angle,mxbai2024card}.
The last is a control: it  advertises int8 and binary use without claiming quantization-aware training.
For each $b \in \{2,\dots,8\}$ we run the full two-round protocol.
Table~\ref{tab:quality-sweep} reports, for $b \in \{2,3,4,8\}$, the fraction of floating-point nDCG@10~\citep{jarvelin2002cumulated,manning2008introduction} each encoder retains.
Appendix~\ref{sec:encoder-suite} gives every bit-width on whole abstracts (Table~\ref{tab:quality-sweep-abstracts}) and top-1 agreement at $b \in \{3,4\}$ (Table~\ref{tab:suite}).
At a high level, the results show that down to $b{=}4$, all encoders retain a high percentage of retrieval quality compared to their floating-point counterparts.

\paragraph{Appropriate quantization.}
A $b$-bit symmetric quantizer maps a chosen magnitude, the scale $\tau$, to the top level $q_{\max} = 2^{b-1}-1$ and rounds every coordinate to the nearest multiple of the step $\tau/q_{\max}$, so every coordinate within $[-\tau, \tau]$ is off by at most half a step.
Max-abs scaling, the usual choice, sets the scale to the largest coordinate of the corpus matrix.
BERT-family encoders carry one coordinate holding over $5\%$ of the squared norm of every embedding, and that one coordinate makes the step up to $4\times$ coarser than the other $1023$ need.
Appendix~\ref{sec:encoder-suite} shows the effect: BGE-large clipped at $b{=}3$ retains $97$ and $99\%$ of floating-point nDCG@10 on SciFact and NFCorpus, where max-abs retains $77$ and $73\%$ at $b{=}3$ and only reaches $97$ and $96\%$ at $b{=}4$.
We thus use clipping quantization for our experiments.

\paragraph{Quality at scale.}
A larger index gives every query more near-misses to lose to~\citep{reimers2021curse}, so we repeat the sweep on Natural Questions from $10^4$ passages to the full $2.68$M with five encoders, one of them trained for low-bit readout (Table~\ref{tab:scale-recall5}, Appendix~\ref{sec:encoder-suite}).
At $b{=}4$ four of the five retain at least $98.8\%$ of floating-point Recall@5 at every size, and e5 retains $97.3\%$ at $2.68$M.
At $b{=}3$ arctic, BGE and mxbai retain at least $98.4\%$ while gte and e5 fall to $94.5$ and $92.0\%$ at $2.68$M.
At $b{=}2$ the loss grows with the corpus, from $1$ to $4\%$ at $10^4$ to $10$ to $23\%$ at $2.68$M.
The lead of the quantization-aware encoder grows with the corpus: arctic retains $0.8$ to $2.7$ points more Recall@5 than the conventional four at $10^4$ and $2.1$ to $12.8$ points more at $2.68$M.
As such, retrieval quality places the operating point at $b{=}3$.
There, five of the six encoders lose $1$ to $3\%$ of nDCG@10 (e5 loses $6\%$ on NFCorpus), and for arctic, BGE and mxbai the loss stays small as the corpus grows.

\subsection{Quantization robustness lowers cryptographic cost}
\label{subsec:correctness}
In this section we formalize how the bit-width $b$ sets the cryptographic parameters, and thus the cost, of Round~1.
For both rounds, we fix the ciphertext modulus $q = 2^{64} - 2^{32} + 1$, at which the primitive of \citet{cryptoeprint:2026/027} is most efficient, and target $\lambda \geq 128$ bits of security and a decryption failure probability of $2^{-64}$ per coefficient.

\begin{theorem}[Correctness of Round~1]
\label{thm:correctness}
Let $\mat{A} \in \ZZ^{m \times N}$ and $\vect{x} \in \ZZ^N$ have entries of absolute value less than $B$, let $p \geq 2NB^2$ and let $\vect{c} \gets \enc(\sk, \vect{x})$ be a GLWE encryption with ciphertext modulus $q$, plaintext modulus $p$ and Gaussian noise of standard deviation $\sigma$.
For every $r > 0$, provided $q/\sigma \geq 2r\sqrt{N}\,B\,p$, each coordinate of $\dec(\sk, \mat{A} \cdot \vect{c})$ equals that of $\mat{A} \cdot \vect{x}$ over the integers, except with probability $\varepsilon(r) \coloneq 2\exp(-r^2/2)/(r\sqrt{2\pi})$.
\end{theorem}

We defer the proof to Appendix~\ref{sec:proof-correctness}.
In Round~1 the inner dimension $N$ is the embedding dimension $d = 2^{10}$, the quantized entries lie in $[-(2^{b-1}-1), 2^{b-1}-1]$ so $B = 2^{b-1}$, and we take $p = 2NB^2 = 2^{2b+9}$.
At $q = 2^{64} - 2^{32} + 1$ and $r = 9.16$, for which $\varepsilon(r) \leq 2^{-64}$, Theorem~\ref{thm:correctness} thus makes Round~1 correct provided $\sigma \leq \sigma_{\max}(b)$ where
\begin{equation}
\label{eq:exchange}
\log_2 \sigma_{\max}(b) \;=\; \log_2 q - \log_2(2r) - \tfrac{1}{2}\log_2 N - (b-1) - (2b+9) \;\approx\; 46.8 - 3b .
\end{equation}
Every added bit of quantization thus costs three bits of noise budget, one through the entry bound $B$ and two through the plaintext modulus $p$.
Security pulls in the other direction and calls for more noise.
It admits no closed form, so we evaluate it with the lattice estimator~\citep{JMC:APS15} at $\sigma = \sigma_{\max}(b)$, for a ternary secret and GLWE dimensions $2^{10}$ and $2^{11}$ (Table~\ref{tab:exchange}).
Dimension $2^{10}$, which matches the embedding dimension, reaches $128$ bits at $b = 2$ only; every $b \geq 3$ must double the GLWE dimension $\kappa d$, by taking $d = 2^{11}$ or rank $\kappa = 2$.
As such, the cryptography admits two regimes rather than a cost per bit: the smaller ring at $b{=}2$, and one doubling, paid once, from $b{=}3$ on.
We measure what this doubling costs in \S\ref{subsec:costs}.

\begin{table}[t]
\centering
\caption{Quantization-cryptography exchange rate: Round~1 parameters at $q = 2^{64} - 2^{32} + 1$, $\kappa = 1$, ternary secret and $2^{-64}$ failure probability. $\lambda_{10}$ and $\lambda_{11}$ are the lattice-estimator security at $\sigma = \sigma_{\max}(b)$ for LWE dimensions $d = 2^{10}$ and $2^{11}$, and $d$ the smallest dimension at $128$ bits.}
\label{tab:exchange}
\small
\begin{tabular}{@{}c cc cc c@{}}
\toprule
$b$ & $p$ & $\sigma_{\max}$ & $\lambda_{10}$ & $\lambda_{11}$ & $d$ \\
\midrule
$\mathbf{2}$ & $2^{13}$ & $2^{40.8}$ & $\mathbf{138}$ & $284$ & $\mathbf{2^{10}}$ \\
$3$ & $2^{15}$ & $2^{37.8}$ & $122$ & $253$ & $2^{11}$ \\
$4$ & $2^{17}$ & $2^{34.8}$ & $110$ & $227$ & $2^{11}$ \\
$5$ & $2^{19}$ & $2^{31.8}$ & $100$ & $206$ & $2^{11}$ \\
$8$ & $2^{25}$ & $2^{22.8}$ & $79$ & $161$ & $2^{11}$ \\
\bottomrule
\end{tabular}
\vspace{-0.5em}
\end{table}

\subsection{Concrete costs of private dense retrieval}
\label{subsec:costs}
In what follows we quantify the cost of a private and verifiable query, isolate the share of that price paid for malicious security and identify the corpus sizes and workflows the protocol can serve today.

\paragraph{Measured cost.}
Table~\ref{tab:benchmarks} reports the cost of the full protocol on the evaluation corpora and at scale, at both parameter sets of \S\ref{subsec:correctness}.
At $b{=}2$ the server time is of $1.7$\,s on SciFact and on NFCorpus, and of $2.5$\,s from $b{=}3$ on, once the ring is doubled.
At $10^6$ chunks the two are of $3.0$ and $5.4$ minutes.
Round~2 multiplies $\mat{D}^\top \in \ZZ_p^{512 \times m}$ where Round~1 multiplies $\mat{A} \in \ZZ_p^{m \times 1{,}024}$.
Both matrices grow linearly in $m$ and the prover's cost follows this trend.
As such, Round~2 costs $\approx 0.5\times$ Round~1 at every corpus size at $b{=}2$.
Per round, the client uploads $\lceil N/d \rceil (\kappa{+}1) d \log_2 q$ bits of ciphertext and downloads $\lceil m'/d \rceil (\kappa{+}1) d \log_2 q$ bits plus the proof, where $m'$ is the output length of that round.
Over both rounds, communication is of $528$~KiB on SciFact and $496$~KiB on NFCorpus at $b{=}2$, proofs included, and $20$~KiB more once the ring is doubled.
Most of the server time is spent computing the proof of correct computation.
On the SciFact Round-1 shape at $d = 2^{10}$, the matrix-vector product over ciphertexts takes $15$\,ms without proof generation and $1.18$\,s with it.
Detecting a malicious server rather than trusting it thus costs a factor of $77$.

\begin{table}[h]
\centering
\caption{Measured cost on the evaluation corpora and at scale, at $b{=}2$ (Round~1 at $d = 2^{10}$) and from $b{=}3$ on (Round~1 at $d = 2^{11}$); Round~2 runs at $d = 2^{12}$ with $\ell = 512$ (Table~\ref{tab:glwe-params}). Server time is measured on $48$ AMD EPYC 9654 cores; client time is measured on a single core. Prices at the on-demand rate of a $48$-core \texttt{hpc7a.24xlarge} EC2 instance (\$$7.20$/h).}
\label{tab:benchmarks}
\small
\setlength{\tabcolsep}{4pt}
\begin{tabular}{@{}lccccc@{}}
\toprule
& & & \multicolumn{3}{c}{$b{=}2$ / $b{\geq}3$} \\
\cmidrule(lr){4-6}
Corpus & Round 1: $\mat{A}$ & Round 2: $\mat{D}^\top$ & Client & Server & \$/query \\
\midrule
NFCorpus ($5{,}761$ chunks) & $5{,}761 \times 1{,}024$ & $512 \times 5{,}761$ & $510$\,ms / $914$\,ms & $1.7$\,s / $2.5$\,s & $0.003$ / $0.005$ \\
SciFact ($7{,}510$ chunks) & $7{,}510 \times 1{,}024$ & $512 \times 7{,}510$ & $510$\,ms / $914$\,ms & $1.7$\,s / $2.5$\,s & $0.003$ / $0.005$ \\
$m = 10^4$ & $10^4 \times 1{,}024$ & $512 \times 10^4$ & $907$\,ms / $1.7$\,s & $3.0$\,s / $5.0$\,s & $0.006$ / $0.010$ \\
$m = 10^5$ & $10^5 \times 1{,}024$ & $512 \times 10^5$ & $6.3$\,s / $12.5$\,s & $21.7$\,s / $33.8$\,s & $0.043$ / $0.068$ \\
$m = 10^6$ & $10^6 \times 1{,}024$ & $512 \times 10^6$ & $50.4$\,s / $1.7$\,min & $3.0$\,min / $5.4$\,min & $0.364$ / $0.650$ \\
\bottomrule
\end{tabular}
\end{table}

\paragraph{The deployable frontier.}
\Cref{tab:benchmarks} presents the cost of a private dense retrieval query under our protocol for both $b \in \{2,3\}$.
By \S\ref{subsec:retrieval}, $b{=}3$ is the smallest bit-width that keeps retrieval quality, so the following assumes $b{=}3$, i.e. the larger GLWE parameters.
Today, UpToDate, the curated reference clinicians consult first, holds over $12{,}000$ topics~\citep{uptodate2026topics}, each the length of a review article, i.e., of the order of $2$ to $5 \times 10^5$ chunks in all.
A private query on a corpus this size costs one to three minutes and is affordable today.
On the other hand, PubMed, the literature behind UpToDate, holds over $40$ million abstracts~\citep{pubmed2026}, one or two chunks each.
A query on such a large corpus would take $4$ to $6$ hours of server time at this throughput and incur communication costs of approximately $1$~GB.
We judge this is too costly for the clinical setting, even in the deferred-response context.
Two things can bring the cost down: faster cryptographic machinery, and encoders that keep their quality at $b{=}2$.
The former lies with the cryptographic community.
The latter is a matter of representation learning, and by \S\ref{subsec:correctness} it is worth the one ring doubling, i.e. the factor of $1.5$ to $1.8$ between the two columns of Table~\ref{tab:benchmarks}.
As such, at $b{=}3$ the protocol serves a corpus of a few hundred thousand chunks in one to three minutes and $10^6$ in five, which covers a curated clinical reference and not the literature behind it.

\subsection{Downstream task quality under quantization}
\label{subsec:downstream}

In the RAG deployments we target, what matters is the accuracy of the answer a language model produces from the retrieved passages rather than the relevance ranking of these passages.
The two need not degrade together, since a reader consuming $k$ passages tolerates a reordering provided the evidence it needs appears somewhere in its context~\citep{salemi2024erag,cuconasu2024noise}.
We measure downstream accuracy as a function of $b$ and locate the smallest bit-width it tolerates.

\paragraph{Methodology.}
We follow the standard end-task protocol for retrieval-augmented generation~\citep{karpukhin2020dense,lewis2020retrieval,petroni2021kilt}: the retriever returns the top-$k$ passages, a fixed reader answers from them, and only the retrieved set changes between the conditions we compare.
A question-answering response is correct if it contains a reference answer~\citep{mallen2023popqa,asai2024selfrag}, and a claim-verification response if it matches the label.
We test each bit-width against floating point with an exact McNemar test on paired per-query outcomes.
Since a non-significant difference is not evidence of equivalence, we further test equivalence within $\pm 5$ points by two one-sided tests (TOST) on the same paired differences (Appendix~\ref{sec:prompts}).

\paragraph{Setup.}
We use three endpoints and two encoders, one from each group of \S\ref{subsec:retrieval}.
\emph{SciFact} is evaluated on its original task: the reader receives a claim and the top-$k$ retrieved abstracts and outputs \textsc{Support}, \textsc{Contradict} or \textsc{NoInfo}, scored against the human labels on the $188$ claims with evidence in the corpus, at $k{=}5$ and $k{=}1$.
\emph{BioASQ} is biomedical factoid QA over a corpus of $35{,}454$ passages, scored by containment of a reference answer.
\emph{Natural Questions} is the out-of-domain control at scale, answered over the full $2.68$M-passage corpus of \S\ref{subsec:retrieval}.
Each endpoint has $300$ queries.
We test the performance limits of the reader by measuring its performance in the closed book setting, where it answers from its own knowledge with no passages, and in the gold context setting, where it receives the labelled relevant passages in place of the retrieved ones.
Retrieval uses the clipped quantizer with BGE and arctic, and only the retrieved set varies with $b$, so any change in accuracy is attributable to quantization.
The reader is \texttt{phi-4} ($14$B)~\citep{abdin2024phi4}.
Appendix~\ref{sec:prompts} repeats every measurement with \texttt{OLMo-2-7B}~\citep{olmo2025olmo2} and \texttt{Qwen2.5} at $7$B and $72$B~\citep{qwen2025qwen25}, and gives the prompts and corpora.

\paragraph{Downstream quality holds to $b{=}3$ and not to $b{=}2$.}
Table~\ref{tab:downstream} shows the result.
On closed book, the reader reaches $19\%$ on BioASQ, $31\%$ on NQ and $59\%$ on SciFact, with gold passages $68$, $74$ and $81\%$, and with floating-point retrieval it sits between the two, close to the gold row.
The effects of quantization below are to be read against that span of $23$ to $49$ points.
No cell shows a significant drop from floating point at $b{=}4$ or $b{=}3$, and every cell is equivalent to floating point within $5$ points at both bit-widths.
This holds at $2.68$M passages.
The reason is that the RAG endpoint consumes recall rather than rank.
More precisely, at $b{=}3$ the top-$5$ set of BGE on NQ agrees with the floating-point one on $82\%$ of positions, yet it contains a relevant passage for $70\%$ of questions, as in floating point.
Further, \S\ref{subsec:retrieval} showed Recall@5 retained at every corpus size.
The three other readers agree: over all four, $32$ of $32$ reader-endpoint-encoder combinations are equivalent to floating point at $b{=}4$ and $30$ of $32$ at $b{=}3$, with a single significant drop of $3.3$ points (Table~\ref{tab:downstream-readers}).

At $b{=}2$, accuracy drops significantly in $20$ of $32$ combinations, by $3$ to $9$ points, and only $9$ remain equivalent to floating point within $5$ points.
Again, poor recall is responsible for this drop in performance.
On NQ, the top-$5$ set of BGE contains a relevant passage for $60\%$ of questions at $b{=}2$ against $70\%$ at $b{=}3$, and the reader answers \textsc{Unknown} more often rather than answering wrongly.
The loss grows with the corpus: Recall@5 retained at $b{=}2$ falls from $98$ to $86\%$ for BGE and from $99$ to $90\%$ for arctic between $10^4$ and $2.68$M passages (Table~\ref{tab:scale-recall5}).
Interestingly, the loss is uneven.
It is largest at $k{=}1$, where a single reordering at the top removes the only passage the reader sees, at $4$ to $9$ points for both encoders.
At $k{=}5$ it depends on the encoder: on SciFact, BGE loses $3$ to $6$ points where arctic loses at most $2.1$, with no significant drop for any of the four readers.
This shows that the loss at $b{=}2$ is a property of how the encoder was trained.

\begin{table}[t]
\centering
\caption{Downstream accuracy versus bit-width with the clipped quantizer (reader \texttt{phi-4}, $n{=}300$ queries per endpoint). SciFact: accuracy on the $188$ evidence-bearing claims; BioASQ and NQ: answer containment, NQ retrieving over the full $2.68$M-passage corpus. Closed book: the reader sees no passages; gold context: it sees the labelled relevant passages (at most $k$). \textbf{Bold} marks a significant drop from floating point (exact McNemar, $p<0.05$). Other readers: Appendix~\ref{sec:prompts}.}
\label{tab:downstream}
\small
\setlength{\tabcolsep}{4pt}
\begin{tabular}{@{}l cc cc cc cc@{}}
\toprule
& \multicolumn{2}{c}{SciFact $k{=}5$} & \multicolumn{2}{c}{SciFact $k{=}1$} & \multicolumn{2}{c}{BioASQ $k{=}5$} & \multicolumn{2}{c}{NQ $k{=}5$, $2.68$M} \\
\cmidrule(lr){2-3}\cmidrule(lr){4-5}\cmidrule(lr){6-7}\cmidrule(lr){8-9}
$b$ & BGE & arctic & BGE & arctic & BGE & arctic & BGE & arctic \\
closed book & \multicolumn{2}{c}{$0.585$} & \multicolumn{2}{c}{$0.585$} & \multicolumn{2}{c}{$0.187$} & \multicolumn{2}{c}{$0.313$} \\
gold context & \multicolumn{2}{c}{$0.814$} & \multicolumn{2}{c}{$0.809$} & \multicolumn{2}{c}{$0.680$} & \multicolumn{2}{c}{$0.740$} \\
\midrule
float & $0.734$ & $0.713$ & $0.766$ & $0.739$ & $0.633$ & $0.647$ & $0.643$ & $0.643$ \\
$8$ & $0.734$ & $0.718$ & $0.771$ & $0.739$ & $0.633$ & $0.643$ & $0.657$ & $0.637$ \\
$6$ & $0.729$ & $0.713$ & $0.771$ & $0.739$ & $0.633$ & $0.640$ & $0.650$ & $0.640$ \\
$4$ & $0.723$ & $0.718$ & $0.771$ & $0.739$ & $0.637$ & $0.637$ & $0.650$ & $0.627$ \\
$3$ & $0.707$ & $0.713$ & $0.761$ & $0.755$ & $0.643$ & $0.640$ & $0.623$ & $0.640$ \\
$2$ & $\mathbf{0.676}$ & $0.697$ & $\mathbf{0.686}$ & $\mathbf{0.681}$ & $0.613$ & $\mathbf{0.607}$ & $\mathbf{0.573}$ & $0.600$ \\
\bottomrule
\end{tabular}
\end{table}

Chunking the SciFact corpus as in \S\ref{sec:experiments} costs the conventional encoder one bit and the quantization-aware one none (Table~\ref{tab:downstream-chunks}, Appendix~\ref{sec:prompts}).
With BGE at $k{=}5$ the loss reaches $4.3$ points at $b{=}3$, where arctic stays within $5$ points of floating point down to $b{=}2$ on this smaller corpus.

As such, downstream accuracy places the operating point of private RAG at $b{=}3$ at every corpus size we measured, one bit above the $b{=}2$ at which \S\ref{subsec:correctness} halves the GLWE parameters required for sufficient security.
That frontier bit is decided by the encoder: with a conventional encoder the loss at $b{=}2$ is of $3$ to $9$ points, with a quantization-aware one it is within $2.1$ points at $k{=}5$ on SciFact.
Whether private RAG runs on the smaller GLWE parameters is thus a property of the encoder.

\section{Discussion}
\label{sec:discussion}

\paragraph{Limitations.}
The protocol is a drop-in retrieval layer for deployments in which a licensed corpus is hosted by an untrusted service, the queries are sensitive and an answer is expected in minutes.
Importantly, our contribution is not an end-to-end privacy-preserving RAG pipeline: we assume the embedding model and the readers run locally or in a trusted environment, and leave a fully private pipeline, which current technology does not yet allow, to future work.
In terms of scale, three constraints bound our protocol.
The first is the client: the Round-1 response, the Round-2 query and the client's decryption and argmax are all linear in $m$.
At $m = 10^6$, Table~\ref{tab:benchmarks} gives one to two minutes of sequential client work and roughly $16$~MB in each direction.
The client is kept single-threaded to model a computationally weak device.
The second is the encoder, which decides whether a task can go below $b{=}3$.
Our evidence is a controlled comparison of six encoders, of which only two were trained for low-bit readout.
What they change is the score margin relative to the rounding step (\S\ref{subsec:retrieval}); whether a training objective can widen it enough for $b{=}2$ at scale is an open problem we hand to representation learning.
Very recent advances in this regard are promising~\citep{3814246.3814371}.
The third constraint is the cryptographic primitive: verifiable HE throughput limits web-scale deployment, but our modular design directly absorbs any improvement to it.

\section{Conclusion}
\label{sec:conclusion}

We have formulated private dense retrieval as a problem of privacy and integrity against a malicious server.
We have reduced the underlying computations to verifiable multiplication of a committed matrix by an encrypted vector and used the resulting protocol to measure what privacy costs.
With a clipped quantizer, retrieval quality and downstream accuracy hold at $b{=}3$ for six encoders and four readers on corpora of up to $2.68$M passages.
A private, verified query over a corpus the size of a clinical reference is then of one to three minutes.
At $128$ bits of security, only $b{=}2$ admits the smaller GLWE parameters; $b \geq 3$ costs a factor of $1.8$ in server time.
Which encoders keep their quality below three bits is set by their score margin relative to the quantization step, and only those trained for low-bit readout do so at $b{=}2$.
Lowering the cost of privacy in dense retrieval is thus not only a cryptography problem, but also one of representation learning. 

\section*{AI Use Statement}
In this work, we used generative AI tools for feedback on research methodology. We have not used generative AI tools to generate synthetic data sets, help develop theoretical models or conceptual frameworks, formulate mathematical claims, provide critical ingredients for proving mathematical claims, assist in the writing of proofs, propose or refine hypotheses, implement methods, assist with translation, clean and reformat datasets, support qualitative and thematic data analysis, or interpret results. 
Additionally, we used generative AI tools to edit software code, generate tables and identify relevant literature. We have reviewed all AI-assisted work. We take responsibility for the final content of this work, including text, claims or artifacts produced with the aid of generative AI.

\section*{Reproducibility Statement}
All retrieval-side results use public corpora (SciFact, NFCorpus, Natural Questions, BioASQ) and publicly available encoders.
The downstream readers are open-weight models run with structured decoding, so every number in \S\ref{subsec:downstream} is reproducible without access to an API.
The quantization code and the per-bit-width sweep outputs are available at \url{\pdrartifacturl}, and the protocol implementation at \url{\vpirurl}.
Prompts are fixed across bit-widths within a reader and are given in Appendix~\ref{sec:prompts}.
The cryptographic parameters, the correctness constraint and the security estimate are stated in closed form in \S\ref{subsec:correctness} and Appendix~\ref{sec:glwe}, so Table~\ref{tab:exchange} is reproducible without running the protocol.
Timings are single-node measurements on a named processor with a fixed thread count, and the benchmark binary asserts verification and exact decryption at every shape.

\section*{Ethics Statement}
This work aims to reduce the exposure of sensitive queries and is motivated by a concrete disclosure risk in clinical decision support.
The corpora we use are public research datasets and contain no patient data.
We emphasize that query privacy against the retrieval provider does not guarantee privacy over the rest of the technological stack.
A protocol such as ours could be invoked to argue compliance while leakage persists elsewhere.
The guarantee is exactly that of Definition~\ref{def:pdr} and no broader.
Further, the integrity guarantee is relative to a committed digest.
Deploying the protocol without independent attestation of that digest by peers or auditors offers the appearance of verification without its substance.

\section*{Acknowledgments}
The authors would like to thank Léo Gagnon, Romane Asselin and Charlie Gauthier for insightful discussions.
This work was supported by the Fonds de recherche du Québec - Nature et technologies (FRQNT) through grant \url{https://doi.org/10.69777/2006324}.


\bibliographystyle{iclr2027_conference}
\bibliography{local,crypto,abbrev3}

\appendix
\section{Preliminaries}
\label{sec:prelim}

\subsection{Homomorphic Encryption}

We rely on a linearly homomorphic encryption scheme based on the General Learning With Errors (GLWE) assumption.
\subsection{(G)LWE encryption}
\label{sec:glwe}
\begin{definition}[GLWE encryption scheme~\cite{DBLP:journals/jacm/LyubashevskyPR13}]
\label{def:glwe}
Let $q$ be a prime ciphertext modulus, $p < q$ a plaintext modulus, $\kappa$ the GLWE dimension parameter and $d$ a power of two ring degree so that $\R_q \coloneq \ZZ[X]/\langle X^d + 1 \rangle$ is a cyclotomic polynomial ring.
We let $\chi_s$ and $\chi_e$ denote the secret key and error distributions over $\R_q$ respectively. Finally, we let $\calU$ denote the uniform distribution, $\R_p$ denote the message space of the encryption scheme and $\Delta \coloneq \lfloor q/p \rceil$.
A GLWE encryption scheme $\calE \coloneq (\kgen, \enc, \dec)$ consists of the following algorithms:
\begin{itemize}
    \item $\kgen(\secparam)$: Takes as input a security parameter $\secparam$ and returns a sampled  secret key $\vec s \sample \chi_s^\kappa$.
    \item $\enc(\vec s, m)$: Takes as input a secret key $\vec s$ and a message $m \in \R_p$ and returns a ciphertext $\vec c \in \R_q^{\kappa+1}$ computed as follows:
    \[
        \vec a \sample \calU(\R_q)^\kappa, \quad e \sample \chi_e, \quad \vec c = (\vec a, \innerp{\vec a}{\vec s} + e + \Delta m)  \in \R_q^{\kappa + 1}.
    \]
    \item $\dec(\vec s, \vec c)$: Takes as input a secret key $\vec s$ and a ciphertext $\vec c = (\vec a, b) \in \R_q^{\kappa + 1}$ and returns a message $m \in \R_p$ computed as follows:
    \[
        m = \left\lfloor \frac{b - \innerp{\vec a}{\vec s}}{\Delta} \right\rceil
    \]
\end{itemize}
This encryption scheme is secure under the GLWE assumption for appropriate choices of parameters.
For a concrete parameter set, its security can be estimated using the widely used lattice estimator~\cite{JMC:APS15}.
\end{definition}

Table~\ref{tab:glwe-params} gives the parameter sets we instantiate the protocol with, satisfying the correctness constraint of Theorem~\ref{thm:correctness} with failure probability at most $2^{-64}$ per coefficient and at least $\lambda = 128$ bits of security for a ternary secret, as estimated with the lattice estimator~\cite{JMC:APS15}.
Round~1 is shown at $b = 2$, the only bit-width the ring $d = 2^{10}$ admits, and at $b = 3$ and $4$ with $d = 2^{11}$. 
The remaining rows of Table~\ref{tab:exchange} follow the same pattern (rank $\kappa = 2$ at $d = 2^{10}$ has the same LWE dimension and the same security, but the implementation we measure supports rank $1$ only).
In each row $\sigma$ is set at the correctness ceiling of \eqref{eq:exchange}, so the security shown is the largest the row allows.
Round~2 encrypts a one-hot vector over $m$ chunks with $B = p = 2^{24}$, so its ceiling is $\log_2 \sigma_{\max} = 64 - \log_2(2r) - \tfrac{1}{2}\log_2 m - 48$, i.e., $2^{5.3}$ at $m = 2^{13}$ and $2^{1.8}$ at $m = 2^{20}$.
We use the minimal width $\sigma = 3.2$, which is correct up to $m = 2^{20}$ chunks and gives $223$ bits at $d = 2^{12}$, the ring at which $128$ bits are reached at that width; beyond $2^{20}$ chunks, each bit removed from $p$ buys four doublings of $m$.

\begin{table}[ht]
\centering
\caption{GLWE parameters per round: ring dimension $d$, rank $\kappa$, ciphertext modulus $q$, plaintext modulus $p$, noise standard deviation $\sigma$ and the resulting security parameter $\lambda$ (ternary secret) and maximum failure probability $\Pr[\mathsf{fail}]$ per coefficient. Round~2 is given for $m \leq 2^{20}$ chunks.}
\label{tab:glwe-params}
\small
\begin{tabular}{@{}lccccccc@{}}
\toprule
Round & $d$ & $\kappa$ & $q$ & $p$ & $\sigma$ & $\lambda$ & $\Pr[\mathsf{fail}]$ \\
\midrule
Round 1 ($b=2$) & $2^{10}$ & $1$ & $2^{64} - 2^{32} + 1$ & $2^{13}$ & $2^{40.8}$ & $138$ & $2^{-64}$ \\
Round 1 ($b=3$) & $2^{11}$ & $1$ & $2^{64} - 2^{32} + 1$ & $2^{15}$ & $2^{37.8}$ & $253$ & $2^{-64}$ \\
Round 1 ($b=4$) & $2^{11}$ & $1$ & $2^{64} - 2^{32} + 1$ & $2^{17}$ & $2^{34.8}$ & $227$ & $2^{-64}$ \\
Round 2 & $2^{12}$ & $1$ & $2^{64} - 2^{32} + 1$ & $2^{24}$ & $2^{1.7}$ & $223$ & $2^{-64}$ \\
\bottomrule
\end{tabular}
\end{table}

The above encryption scheme is additively homomorphic and allows for multiplication by plaintext elements,
provided the noise remains ``small''.
More precisely,
given two ciphertexts $\vec c_1 = (\vec a_1, b_1)$ and $\vec c_2 = (\vec a_2, b_2)$ encrypting messages $m_1, m_2 \in \R_p$ respectively,
we have:
\begin{itemize}
    \item Addition: $\vec c_{\text{add}} = \vec c_1 + \vec c_2 = (\vec a_1 + \vec a_2, b_1 + b_2)$ is a ciphertext encrypting $m_{\text{add}} = m_1 + m_2$.
    \item Scalar multiplication: For any plaintext $y \in \R_p$, $\vec c_{\text{scal}} = y \cdot \vec c_1 = (y \cdot \vec a_1, y \cdot b_1)$ is a ciphertext encrypting $m_{\text{scal}} = y \cdot m_1$.
\end{itemize}
These properties hold provided the accumulated noise in the resulting ciphertexts remains small enough to allow for correct decryption,
i.e. $\inftynorm{e} \leq \Delta/2$.

\subsection{Proof of Theorem~\ref{thm:correctness}}
\label{sec:proof-correctness}
\begin{proof}
Each score $s_i = \sum_j A_{ij} x_j$ satisfies $|s_i| \leq N(B-1)^2 < NB^2$, so the $2NB^2 - 1 \leq p$ values it can take are distinct modulo $p$.
It remains to show that decryption returns $s_i \bmod p$.
By linearity, $\mat{A} \cdot \vect{c}$ decrypts to $\mat{A} \cdot \vect{x} \bmod p$ provided the noise of every coefficient stays below $\Delta/2$ in absolute value.
We neglect the rounding of $\Delta = \lfloor q/p \rceil$, which shifts this threshold by at most $p/2$, and write $\Delta/2 = q/2p$.
The ciphertext $\vect{c}$ carries $N$ independent noise coefficients $e_1, \dots, e_N \sim \mathcal{N}(0, \sigma^2)$, and the coefficient of $\mat{A} \cdot \vect{c}$ carrying $s_i$ has noise $e'_i = \sum_j \pm A_{ij}\, e_{\pi(j)}$, for a permutation $\pi$ and signs fixed by the negacyclic ring structure.
As such, $e'_i$ is Gaussian of variance $\sigma^2 \sum_j A_{ij}^2 \leq NB^2\sigma^2$, and the Gaussian tail bound gives $\Pr[|e'_i| > r\sqrt{N}B\sigma] \leq \varepsilon(r)$.
The inequality $q/\sigma \geq 2r\sqrt{N}Bp$ implies that $r\sqrt{N}B\sigma \leq q/(2p)$, so every coordinate is correct except with probability $\varepsilon(r)$.
\end{proof}

\subsection{Secure Matrix-Vector Multiplication}
\label{sec:secure-matvec}

We assume the existence of an efficient protocol for verifiable plaintext matrix--encrypted vector multiplication.
Such a primitive can be instantiated with e.g.~\cite{cryptoeprint:2026/027}.
The protocol operates as follows:

\begin{enumerate}
    \item The client encodes its input $\vect{x} = (x_1, \ldots, x_N) \in \ZZ_p^N$ into ring elements and encrypts each component, obtaining a ciphertext vector $\vect{c}$ which is sent to the server.
    \item The server holds a plaintext matrix $\mat{M} \in \ZZ_p^{m \times N}$. It computes the matrix-ciphertext product row-wise, yielding a result ciphertext vector $\vect{c}'$ of dimension $m$, which is returned to the client.
    \item The client decrypts $\vect{c}'$ to obtain $\mat{M} \cdot \vect{x}$.
\end{enumerate}

\begin{definition}[Passive security]
The protocol is \emph{passively secure} if the server's view (the ciphertext vector $\vect{c}$) reveals no information about $\vect{x}$ beyond public parameters and dimensions.
\end{definition}

\begin{definition}[Malicious security / Verifiability]
The protocol is \emph{maliciously secure} if, in addition to passive security, the client can verify that the server computed $\mat{M} \cdot \vect{x}$ correctly. That is, a malicious server cannot convince the client to accept an incorrect result $\vect{y} \neq \mat{M} \cdot \vect{x}$ except with probability $\negl$.
\end{definition}

The malicious security property is achieved via a verifiable computation mechanism
which allows the client to check the server's computation without re-doing it.

\subsection{Embedding-Based Retrieval}
\label{sec:embeddings}

Modern retrieval systems represent documents and queries as dense vectors (embeddings) in $\mathbb{R}^n$
using neural network encoders (\cite{reimers2019sentence, ni2022large}).
Given a corpus of $m$ documents, each document $i$ is mapped to an embedding $\vect{a}_i \in \mathbb{R}^n$.
These embeddings are collected as rows of a matrix $\mat{A} \in \mathbb{R}^{m \times n}$.
For a query with embedding $\vect{x} \in \mathbb{R}^n$,
the relevance of each document is measured by the inner product $\langle \vect{a}_i, \vect{x} \rangle$,
and the top-$k$ documents with highest scores are retrieved.

\subsection{Quantized Embeddings}
\label{sec:quantization}

To operate over $\ZZ_p$ rather than $\mathbb{R}$ we use \emph{clipped symmetric scalar quantization}.
Let $q_{\max} = 2^{b-1}-1$ and let $\tau$ be the $99$th percentile of $|\mat{A}|$ over all entries of the embedding matrix.
Each coordinate $a$ is mapped to $\mathrm{clip}(\lfloor a\,q_{\max}/\tau \rceil, -q_{\max}, q_{\max})$, so the $1\%$ of coordinates beyond $\tau$ saturate.
A query vector is quantized the same way with its own percentile.
Every entry lies in $[-q_{\max}, q_{\max}]$.

The usual max-abs rule maps the single largest coordinate to $q_{\max}$.
BERT-family encoders carry one coordinate holding $5$ to $6\%$ of the squared norm of every embedding, so under max-abs that coordinate sets a step $4\times$ coarser than the bulk of the coordinates need.
Clipping at a percentile is standard in post-training quantization~\citep{banner2019posttraining}.
 Table~\ref{tab:quantizer} compares clipping with max-abs on every encoder.
The $99.9$th and $95$th percentiles land within a point of the $99$th on both corpora.

\paragraph{Compatibility with secure computation.}
Quantized embeddings are natively integers in $\ZZ_p$, requiring no floating-point emulation, and the inner product of quantized vectors is computed exactly by the protocol.
How closely it preserves the floating-point ranking is the subject of \S\ref{subsec:retrieval} and \S\ref{subsec:costs}.

\section{Protocol Message Flow}
\label{sec:protofig}

Figure~\ref{fig:protocol} gives the message flow of the two-round protocol of \S\ref{sec:protocol}.

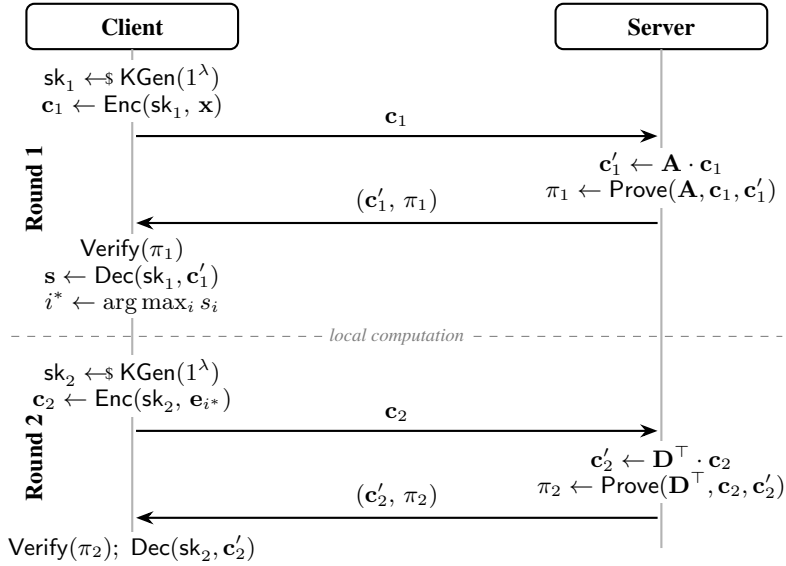
\begin{figure}[ht]
\centering
\begin{tikzpicture}[
  >=Stealth,
  party/.style={rectangle, draw, thick, rounded corners=3pt,
                minimum width=2.8cm, minimum height=0.6cm,
                font=\bfseries\small},
  msg/.style={font=\footnotesize},
  comp/.style={font=\footnotesize\itshape, align=center, inner sep=1.5pt, fill=white},
  timeline/.style={thick, gray!60}
]

\node[party] (C) at (0,0) {Client};
\node[party] (S) at (7,0) {Server};

\draw[timeline] (0,-0.3) -- (0,-6.9);
\draw[timeline] (7,-0.3) -- (7,-6.9);

\node[comp] at (0,-0.85) {$\sk_1 \sample \kgen(\secparam)$\\$\vect{c}_1 \gets \enc(\sk_1,\, \vect{x})$};
\draw[->, thick] (0.05,-1.45) -- (6.95,-1.45)
  node[midway, above, msg] {$\vect{c}_1$};

\node[comp] at (7,-2.0) {$\vect{c}'_1 \gets \mat{A} \cdot \vect{c}_1$\\$\pi_1 \gets \mathsf{Prove}(\mat{A}, \vect{c}_1, \vect{c}'_1)$};
\draw[->, thick] (6.95,-2.6) -- (0.05,-2.6)
  node[midway, above, msg] {$(\vect{c}'_1,\, \pi_1)$};

\node[comp] at (0,-3.3) {$\mathsf{Verify}(\pi_1)$\\$\vect{s} \gets \dec(\sk_1, \vect{c}'_1)$\\$i^* \gets \arg\max_i s_i$};

\draw[dashed, gray] (-1.6,-4.1) -- (8.6,-4.1);
\node[font=\scriptsize\itshape, gray, fill=white, inner sep=1.5pt] at (3.5,-4.1) {local computation};

\node[comp] at (0,-4.75) {$\sk_2 \sample \kgen(\secparam)$\\$\vect{c}_2 \gets \enc(\sk_2,\, \vect{e}_{i^*})$};
\draw[->, thick] (0.05,-5.35) -- (6.95,-5.35)
  node[midway, above, msg] {$\vect{c}_2$};

\node[comp] at (7,-5.9) {$\vect{c}'_2 \gets \mat{D}^\top \cdot \vect{c}_2$\\$\pi_2 \gets \mathsf{Prove}(\mat{D}^\top, \vect{c}_2, \vect{c}'_2)$};
\draw[->, thick] (6.95,-6.5) -- (0.05,-6.5)
  node[midway, above, msg] {$(\vect{c}'_2,\, \pi_2)$};

\node[comp] at (0,-6.9) {$\mathsf{Verify}(\pi_2);\; \dec(\sk_2, \vect{c}'_2)$};

\node[font=\small\bfseries, rotate=90] at (-1.3,-2.15) {Round 1};
\node[font=\small\bfseries, rotate=90] at (-1.3,-5.7) {Round 2};

\end{tikzpicture}
\caption{The private dense retrieval protocol. Round~1 retrieves encrypted similarity scores, Round~2 privately fetches the top document; both use the same verifiable matrix-vector sub-protocol, and the client verifies each proof before decrypting.}
\label{fig:protocol}
\end{figure}

\section{Security of the Protocol}
\label{sec:proof}

We sketch the proof of Theorem~\ref{thm:security}.

\paragraph{Client privacy.}
In each round the server observes only fresh ciphertexts, which by semantic security of the GLWE scheme are computationally indistinguishable from encryptions of any other vector of the same dimension.
The two rounds use independent keys, so there is no cross-round linkage: the server's view in the two rounds is simulatable from public parameters alone.
Because the client verifies each proof before decrypting, its observable behavior (continue or abort) is a function of the proof's validity only, which the server can compute itself.
As such, the client's reactions leak nothing about the decrypted values, which rules out reaction attacks~\cite{cryptoeprint:2016/1164}.

\paragraph{Integrity.}
By soundness of the underlying protocol, any server deviation in round $i$, be it substituting the matrix, corrupting the ciphertext or fabricating the proof, causes $\pi_i$ to fail verification except with negligible probability:
\[
  \Pr\bigl[\mathsf{Verify}(\pi_i) = \acc \;\wedge\; \vect{y}_i \neq \mat{M}_i \cdot \vect{x}_i\bigr] \leq \negl,
\]
where $(\mat{M}_1, \vect{x}_1) = (\mat{A}, \vect{x})$ and $(\mat{M}_2, \vect{x}_2) = (\mat{D}^\top, \vect{e}_{i^*})$.
Each round is verified against its own commitment, so a union bound over the two rounds gives the claim for the composed protocol.

\section{Batched Top-$k$ Retrieval}
\label{sec:topk-batch}

After Round~1, the client can identify the top-$k$ indices $i_1^*, \ldots, i_k^*$ and issue $k$ independent Round-2 queries.
Using a single $k$-hot vector instead does not work: $\mat{D}^\top \cdot \sum_j \vect{e}_{i_j^*}$ yields the \emph{sum} of the $k$ documents, from which the individual documents cannot be recovered.
For moderate $k$, the $k$ queries batch into a single matrix--matrix product $\mat{D}^\top \cdot \mat{E}$, where $\mat{E} \in \{0,1\}^{m \times k}$ has one-hot columns.

\paragraph{Why latency does not grow with $k$.}
The verifiable matrix--vector multiplication of~\cite{cryptoeprint:2026/027} proves the relation $\vect{c}' = \mat{M} \cdot \vect{c}$ with a sum-check protocol compiled with a polynomial commitment scheme (PCS).
Once per database, the server commits to a polynomial encoding of $\mat{M}$.
To answer a query, it runs the sum-check, which reduces the claim $\vect{c}' = \mat{M} \cdot \vect{c}$ to the evaluation of the committed polynomial at a single random point fixed during the protocol, and then produces a PCS evaluation proof for that point.
The evaluation proof is the dominant cost factor: it accounts for over $95\%$ of the server time.

Batching exploits the fact that the sum-check protocol is linear in the claim it proves.
More precisely, given $k$ claims $\vect{c}'_j = \mat{M} \cdot \vect{c}_j$ about the same committed matrix, the verifier samples a random challenge $\gamma$ and the prover proves the single claim $\sum_j \gamma^j \vect{c}'_j = \mat{M} \cdot \sum_j \gamma^j \vect{c}_j$ instead.
If any of the $k$ claims is false, the combined claim is false as well except with probability at most $k/|\mathcal{C}|$ over the choice of $\gamma$, for $\mathcal{C}$ the challenge space.
This additive soundness loss is negligible provided $|\mathcal{C}|$ is large enough.
The batched claim is a matrix--vector relation of the same shape as a single query.
As such, it is proved with one run of the sum-check and one PCS evaluation proof, regardless of $k$.
The server still computes the $k$ products and the sum-check messages of the combined claim, whose cost grows with $k$, but the evaluation proof, which dominates, is generated once.
The $k$ answers are all returned since the client needs each document, which is why communication grows by a factor $k$ whereas the proof is essentially of the same size as for a single query.

\section{Readers, Prompts and Downstream Corpora}
\label{sec:prompts}

\paragraph{Protocol description.}
We evaluate every endpoint, encoder and condition $c \in \{\text{closed}, \text{gold}, \text{float}, 8, 6, 4, 3, 2\}$.
The evaluation of \S\ref{subsec:downstream} works as follows.
(1)~The corpus and the $300$ queries are embedded once in floating point.
For a bit-width $c = b$, both are quantized with the clipped quantizer of Appendix~\ref{sec:quantization}.
They are then scored in integers, as in Round~1, and the top-$k$ passages are kept.
Ties are broken toward the lower index.
The float condition uses the floating-point scores.
The gold condition uses the labelled relevant passages, at most $k$ of them.
The closed condition uses no passage.
(2)~We build one prompt from the template of the endpoint (below), with the passages in rank order.
A prompt that is identical across conditions is answered once.
(3)~The reader answers with greedy decoding.
Decoding is constrained to the three labels for claim verification and capped at $24$ tokens for QA.
(4)~A QA answer is correct if the normalised text of any reference answer appears in the normalised response.
A claim label is correct if it equals the human label.
SciFact accuracy is computed over the $188$ claims with evidence in the corpus.
(5)~We compare per-query correctness under each condition with the float condition on the same queries, using the tests described below.
We record retrieval-side statistics alongside, namely top-$k$ overlap with float and whether a relevant passage is in the top $k$.
The prompts and $k$ were fixed before any bit-width was run. Nothing is tuned on the test queries.

\paragraph{Readers.}
The body reports \texttt{phi-4} ($14$B, MIT licence)~\citep{abdin2024phi4}.
Table~\ref{tab:downstream-readers} adds \texttt{OLMo-2-7B-Instruct}~\citep{olmo2025olmo2}, whose training data is fully open, and the Qwen2.5 instruction-tuned models at $7$B and $72$B~\citep{qwen2025qwen25}.
All are run locally with vLLM and greedy decoding.
Under the one-word prompt, \texttt{OLMo-2} mostly abstains on claim verification.
It answers \textsc{NoInfo} on $87\%$ of SciFact prompts and never \textsc{Contradict}, so its absolute SciFact accuracy is low.
It still reproduces the shape of the degradation curve, which is what we compare across readers.

\paragraph{Prompts.}
The claim-verification prompt states the claim and lists the retrieved abstracts.
It asks for exactly one word from \{\textsc{Support}, \textsc{Contradict}, \textsc{NoInfo}\}.
Decoding is constrained to those three strings. As such, there are no parse failures and no exclusions.
The QA prompt states the question and lists the retrieved passages.
It asks for the answer in a few words, or for exactly \textsc{Unknown} if the passages do not contain it.
Responses are capped at $24$ tokens.
Prompts are fixed across bit-widths and encoders. Only the retrieved passages change.
One SciFact $k{=}5$ prompt exceeds \texttt{OLMo-2}'s $4{,}096$-token context and is excluded for that reader.

\paragraph{BioASQ corpus.}
We build the corpus from BioASQ-QA~\citep{krithara2023bioasq} (\texttt{training11b}, $4{,}719$ questions).
For every question of every type, we group the evidence snippets by source document.
We concatenate them into one passage per (question, document) pair.
This gives $35{,}454$ passages with a median length of $220$ characters.
Queries are a seeded sample of $300$ factoid questions with a non-empty \texttt{exact\_answer} and at most three gold passages.
Of these, $190$ have one gold passage, $54$ have two and $56$ have three.
The cap keeps gold-in-top-$k$ from being trivially high.
The relevant set of a question is its own passages.
A response is correct if the normalised text of any reference answer synonym appears in it.

\paragraph{NQ at scale.}
The $300$ Natural Questions test queries are the $100$ of the original evaluation plus $200$ more.
These $200$ are a seeded draw among the queries whose question matches an NQ-open answer list.
Passages are retrieved from the full BEIR corpus of $2{,}681{,}468$ Wikipedia passages.
We use the same embeddings as the scale run of Appendix~\ref{sec:encoder-suite}.

\paragraph{Statistics.}
For every (reader, endpoint, encoder) and bit-width, we report the paired difference in accuracy from floating point, with a normal-approximation $95\%$ interval on the per-query differences.
Table~\ref{tab:downstream-readers} gives these differences at $b \in \{4,3,2\}$.
We test for a drop with an exact McNemar test on the discordant pairs.
A non-significant McNemar test is not evidence of equivalence.
As such, we also run two one-sided tests (TOST) with equivalence margins of $5$ and $3$ points.
A bit-width is declared equivalent to floating point when both one-sided tests reject at $p<0.05$, i.e. when the $90\%$ interval lies inside the margin.
At $b{=}4$, all $32$ combinations are equivalent within $5$ points and $23$ are equivalent within $3$.
At $b{=}3$, these counts are $30$ and $17$. At $b{=}2$, they are $9$ and $2$.

\begin{table}[ht]
\centering
\caption{Paired accuracy difference from floating point, in points with $95\%$ interval, at $b{=}4$, $3$ and $2$. All readers, endpoints and encoders use the clipped quantizer ($n{=}300$ queries, $188$ evidence-bearing claims for SciFact). $^{*}$: significant drop, exact McNemar $p<0.05$. $^{\dagger}$: equivalent to floating point within $\pm 5$ points, TOST $p<0.05$. ``float'' is the floating-point accuracy.}
\label{tab:downstream-readers}
\footnotesize
\setlength{\tabcolsep}{3.5pt}
\begin{tabular}{@{}llc lll@{}}
\toprule
Endpoint & Encoder & float & $b{=}4$ & $b{=}3$ & $b{=}2$ \\
\midrule
\multicolumn{6}{@{}l}{\emph{\texttt{phi-4} (14B)}}\\
SciFact $k{=}5$ & BGE & $0.734$ & $-1.1$ $[-3.6, +1.5]$$^{\dagger}$ & $-2.7$ $[-5.4, +0.1]$$^{\dagger}$ & $-5.9$ $[-10.1, -1.6]$$^{*}$ \\
SciFact $k{=}5$ & arctic & $0.713$ & $+0.5$ $[-1.3, +2.3]$$^{\dagger}$ & $+0.0$ $[-2.6, +2.6]$$^{\dagger}$ & $-1.6$ $[-4.7, +1.5]$$^{\dagger}$ \\
SciFact $k{=}1$ & BGE & $0.766$ & $+0.5$ $[-0.5, +1.6]$$^{\dagger}$ & $-0.5$ $[-2.3, +1.3]$$^{\dagger}$ & $-8.0$ $[-12.1, -3.8]$$^{*}$ \\
SciFact $k{=}1$ & arctic & $0.739$ & $+0.0$ $[+0.0, +0.0]$$^{\dagger}$ & $+1.6$ $[-0.2, +3.4]$$^{\dagger}$ & $-5.9$ $[-10.1, -1.6]$$^{*}$ \\
BioASQ & BGE & $0.633$ & $+0.3$ $[-1.4, +2.1]$$^{\dagger}$ & $+1.0$ $[-0.7, +2.7]$$^{\dagger}$ & $-2.0$ $[-4.6, +0.6]$$^{\dagger}$ \\
BioASQ & arctic & $0.647$ & $-1.0$ $[-2.5, +0.5]$$^{\dagger}$ & $-0.7$ $[-2.3, +0.9]$$^{\dagger}$ & $-4.0$ $[-6.6, -1.4]$$^{*}$ \\
NQ ($2.68$M) & BGE & $0.643$ & $+0.7$ $[-2.1, +3.4]$$^{\dagger}$ & $-2.0$ $[-5.3, +1.3]$$^{\dagger}$ & $-7.0$ $[-11.5, -2.5]$$^{*}$ \\
NQ ($2.68$M) & arctic & $0.643$ & $-1.7$ $[-3.6, +0.3]$$^{\dagger}$ & $-0.3$ $[-2.9, +2.2]$$^{\dagger}$ & $-4.3$ $[-8.4, -0.3]$ \\
\midrule
\multicolumn{6}{@{}l}{\emph{\texttt{OLMo-2} (7B)}}\\
SciFact $k{=}5$ & BGE & $0.213$ & $-0.5$ $[-3.3, +2.2]$$^{\dagger}$ & $-1.6$ $[-3.9, +0.7]$$^{\dagger}$ & $-6.4$ $[-10.7, -2.0]$$^{*}$ \\
SciFact $k{=}5$ & arctic & $0.198$ & $-1.1$ $[-3.2, +1.0]$$^{\dagger}$ & $-2.1$ $[-4.7, +0.4]$$^{\dagger}$ & $-2.1$ $[-5.4, +1.2]$$^{\dagger}$ \\
SciFact $k{=}1$ & BGE & $0.505$ & $+1.1$ $[-0.4, +2.5]$$^{\dagger}$ & $-1.1$ $[-3.1, +1.0]$$^{\dagger}$ & $-4.8$ $[-8.5, -1.1]$$^{*}$ \\
SciFact $k{=}1$ & arctic & $0.479$ & $+0.0$ $[-1.5, +1.5]$$^{\dagger}$ & $+0.5$ $[-1.3, +2.3]$$^{\dagger}$ & $-3.7$ $[-7.1, -0.3]$ \\
BioASQ & BGE & $0.503$ & $+1.0$ $[-1.5, +3.5]$$^{\dagger}$ & $+0.0$ $[-3.1, +3.1]$$^{\dagger}$ & $-5.0$ $[-9.3, -0.7]$$^{*}$ \\
BioASQ & arctic & $0.520$ & $-0.3$ $[-2.9, +2.2]$$^{\dagger}$ & $-1.3$ $[-3.9, +1.3]$$^{\dagger}$ & $-2.7$ $[-6.0, +0.7]$ \\
NQ ($2.68$M) & BGE & $0.497$ & $+1.0$ $[-2.1, +4.1]$$^{\dagger}$ & $-1.3$ $[-4.9, +2.2]$$^{\dagger}$ & $-5.0$ $[-9.5, -0.5]$$^{*}$ \\
NQ ($2.68$M) & arctic & $0.563$ & $-0.3$ $[-2.5, +1.8]$$^{\dagger}$ & $-3.3$ $[-6.2, -0.4]$$^{*}$ & $-8.3$ $[-12.2, -4.5]$$^{*}$ \\
\midrule
\multicolumn{6}{@{}l}{\emph{\texttt{Qwen2.5-7B}}}\\
SciFact $k{=}5$ & BGE & $0.681$ & $-0.5$ $[-3.7, +2.6]$$^{\dagger}$ & $-1.6$ $[-4.7, +1.5]$$^{\dagger}$ & $-5.3$ $[-9.9, -0.7]$$^{*}$ \\
SciFact $k{=}5$ & arctic & $0.654$ & $-0.5$ $[-2.3, +1.3]$$^{\dagger}$ & $+1.1$ $[-1.5, +3.6]$$^{\dagger}$ & $+0.0$ $[-3.3, +3.3]$$^{\dagger}$ \\
SciFact $k{=}1$ & BGE & $0.665$ & $+0.5$ $[-0.5, +1.6]$$^{\dagger}$ & $-1.1$ $[-3.1, +1.0]$$^{\dagger}$ & $-9.0$ $[-13.4, -4.7]$$^{*}$ \\
SciFact $k{=}1$ & arctic & $0.644$ & $-0.5$ $[-1.6, +0.5]$$^{\dagger}$ & $+0.0$ $[-2.1, +2.1]$$^{\dagger}$ & $-5.9$ $[-10.3, -1.4]$$^{*}$ \\
BioASQ & BGE & $0.477$ & $-1.7$ $[-4.0, +0.7]$$^{\dagger}$ & $-1.3$ $[-3.9, +1.3]$$^{\dagger}$ & $-4.3$ $[-7.8, -0.8]$$^{*}$ \\
BioASQ & arctic & $0.470$ & $-0.3$ $[-2.3, +1.6]$$^{\dagger}$ & $-1.0$ $[-3.7, +1.7]$$^{\dagger}$ & $-1.7$ $[-4.8, +1.5]$$^{\dagger}$ \\
NQ ($2.68$M) & BGE & $0.353$ & $+1.7$ $[-1.0, +4.4]$$^{\dagger}$ & $+1.7$ $[-1.6, +4.9]$$^{\dagger}$ & $-1.3$ $[-5.6, +2.9]$$^{\dagger}$ \\
NQ ($2.68$M) & arctic & $0.413$ & $-1.7$ $[-3.8, +0.5]$$^{\dagger}$ & $-2.3$ $[-5.2, +0.5]$$^{\dagger}$ & $-5.3$ $[-9.4, -1.2]$$^{*}$ \\
\midrule
\multicolumn{6}{@{}l}{\emph{\texttt{Qwen2.5-72B}}}\\
SciFact $k{=}5$ & BGE & $0.830$ & $+0.5$ $[-1.3, +2.3]$$^{\dagger}$ & $+0.5$ $[-1.3, +2.3]$$^{\dagger}$ & $-2.7$ $[-5.4, +0.1]$$^{\dagger}$ \\
SciFact $k{=}5$ & arctic & $0.819$ & $+0.0$ $[+0.0, +0.0]$$^{\dagger}$ & $+0.5$ $[-1.3, +2.3]$$^{\dagger}$ & $-0.5$ $[-3.3, +2.2]$$^{\dagger}$ \\
SciFact $k{=}1$ & BGE & $0.830$ & $+0.5$ $[-0.5, +1.6]$$^{\dagger}$ & $-1.1$ $[-3.1, +1.0]$$^{\dagger}$ & $-9.0$ $[-13.7, -4.4]$$^{*}$ \\
SciFact $k{=}1$ & arctic & $0.793$ & $+0.0$ $[+0.0, +0.0]$$^{\dagger}$ & $+0.0$ $[-2.6, +2.6]$$^{\dagger}$ & $-7.4$ $[-12.0, -2.9]$$^{*}$ \\
BioASQ & BGE & $0.580$ & $+0.0$ $[-0.9, +0.9]$$^{\dagger}$ & $+0.7$ $[-0.9, +2.3]$$^{\dagger}$ & $-1.7$ $[-4.5, +1.2]$$^{\dagger}$ \\
BioASQ & arctic & $0.597$ & $+0.0$ $[-1.3, +1.3]$$^{\dagger}$ & $-0.3$ $[-1.5, +0.8]$$^{\dagger}$ & $-3.3$ $[-6.1, -0.6]$$^{*}$ \\
NQ ($2.68$M) & BGE & $0.573$ & $-1.7$ $[-4.0, +0.7]$$^{\dagger}$ & $-2.7$ $[-5.7, +0.4]$ & $-6.7$ $[-10.5, -2.8]$$^{*}$ \\
NQ ($2.68$M) & arctic & $0.567$ & $+0.3$ $[-1.4, +2.1]$$^{\dagger}$ & $+0.0$ $[-2.3, +2.3]$$^{\dagger}$ & $-5.0$ $[-8.7, -1.3]$$^{*}$ \\
\bottomrule
\end{tabular}
\end{table}

\paragraph{Chunked SciFact.}
Table~\ref{tab:downstream-chunks} gives the SciFact endpoint for \texttt{phi-4} when the reader receives the top-$k$ chunks of \S\ref{sec:experiments} instead of whole abstracts.
Chunks are cut at sentence boundaries and hold at most $1{,}536$ bytes including the title (\texttt{implem/scale/chunk\_corpus.py}).
The $5{,}183$ abstracts give $7{,}510$ chunks of $1{,}070$ bytes on average.
In floating point, a chunk of a gold abstract is in the top $5$ for $79\%$ of claims with BGE and $77\%$ with arctic.
For whole abstracts, these figures are $80$ and $79\%$.
NFCorpus, chunked the same way, gives $5{,}761$ chunks, and $44\%$ of its abstracts fit in one.
In floating point, the reader does better on five chunks than on five abstracts, by $4.3$ points with BGE and $2.1$ with arctic.
At $k{=}1$ it does about one point worse, since a single chunk may leave out the evidence sentence.
Under quantization, the two encoders separate one bit earlier than on abstracts.
The reason is that chunks of one abstract are near-duplicates, and quantization reorders them.
Judged per document, the ranking survives exactly as on abstracts (Table~\ref{tab:quality-sweep}).
The reader, judged per chunk, achieves top-1 agreement at $b{=}3$ is $0.82$ for BGE and $0.89$ for arctic, against $0.87$ and $0.93$ per document (Appendix~\ref{sec:encoder-suite}).
With BGE at $k{=}5$, the loss reaches $4.3$ points at $b{=}3$ ($p = 0.008$).
On abstracts, it was $2.7$ points and not significant.
With arctic at $k{=}5$, every bit-width down to $b{=}2$ stays within $5$ points of floating point, as on abstracts.
As such, chunking moves the operating point of the conventional encoder from $b{=}3$ to $b{=}4$.
The quantization-aware encoder stays at $b{=}2$.

\begin{table}[ht]
\centering
\caption{SciFact with chunks against whole abstracts, reader \texttt{phi-4}. We give floating-point accuracy on the $188$ evidence-bearing claims. We also give the paired difference from floating point at $b{=}4$, $3$ and $2$, in points with $95\%$ interval. $^{*}$: significant drop, exact McNemar $p<0.05$. $^{\dagger}$: equivalent within $\pm 5$ points, TOST $p<0.05$.}
\label{tab:downstream-chunks}
\footnotesize
\setlength{\tabcolsep}{3.5pt}
\begin{tabular}{@{}lllc lll@{}}
\toprule
Endpoint & Encoder & Unit & float & $b{=}4$ & $b{=}3$ & $b{=}2$ \\
\midrule
SciFact $k{=}5$ & BGE & abstracts & $0.734$ & $-1.1$ $[-3.6, +1.5]$$^{\dagger}$ & $-2.7$ $[-5.4, +0.1]$$^{\dagger}$ & $-5.9$ $[-10.1, -1.6]$$^{*}$ \\
 & & chunks & $0.777$ & $-2.7$ $[-5.0, -0.4]$$^{\dagger}$ & $-4.3$ $[-7.1, -1.4]$$^{*}$ & $-6.9$ $[-11.4, -2.5]$$^{*}$ \\
SciFact $k{=}5$ & arctic & abstracts & $0.713$ & $+0.5$ $[-1.3, +2.3]$$^{\dagger}$ & $+0.0$ $[-2.6, +2.6]$$^{\dagger}$ & $-1.6$ $[-4.7, +1.5]$$^{\dagger}$ \\
 & & chunks & $0.734$ & $+1.6$ $[-0.7, +3.9]$$^{\dagger}$ & $+1.6$ $[-1.9, +5.1]$$^{\dagger}$ & $-1.1$ $[-5.2, +3.1]$$^{\dagger}$ \\
SciFact $k{=}1$ & BGE & abstracts & $0.766$ & $+0.5$ $[-0.5, +1.6]$$^{\dagger}$ & $-0.5$ $[-2.3, +1.3]$$^{\dagger}$ & $-8.0$ $[-12.1, -3.8]$$^{*}$ \\
 & & chunks & $0.755$ & $-0.5$ $[-2.3, +1.3]$$^{\dagger}$ & $-2.1$ $[-4.7, +0.4]$$^{\dagger}$ & $-12.2$ $[-17.6, -6.9]$$^{*}$ \\
SciFact $k{=}1$ & arctic & abstracts & $0.739$ & $+0.0$ $[+0.0, +0.0]$$^{\dagger}$ & $+1.6$ $[-0.2, +3.4]$$^{\dagger}$ & $-5.9$ $[-10.1, -1.6]$$^{*}$ \\
 & & chunks & $0.729$ & $+0.5$ $[-1.8, +2.9]$$^{\dagger}$ & $+0.0$ $[-3.0, +3.0]$$^{\dagger}$ & $-5.9$ $[-10.6, -1.1]$$^{*}$ \\
\bottomrule
\end{tabular}
\end{table}

\section{The Encoder Suite and the Scale Run}
\label{sec:encoder-suite}

In this appendix we give the retrieval-side numbers behind \S\ref{subsec:retrieval} and \S\ref{subsec:costs} for every encoder.
We also compare them with the max-abs quantizer and explain why some encoders lose more quality than others.

\paragraph{Encoders and prefixes.}
Each encoder is run with the instruction prefix its model card prescribes.
\texttt{bge-large-en-v1.5} and \texttt{mxbai-embed-large-v1} prepend ``Represent this sentence for searching relevant passages: '' to queries.
\texttt{e5-large-v2} prepends ``query: '' and ``passage: ''.
\texttt{arctic-embed-l-v2.0} prepends ``query: ''.
\texttt{gte-large} uses no prefix, and Cohere takes an input type instead.
We run \texttt{arctic-embed-l-v2.0} with a $512$-token window instead of its native $8{,}192$, to match the other encoders.
Abstracts rarely exceed this window.
The downstream evaluation of \S\ref{subsec:downstream} uses the clipped quantizer and these prefixes.

\paragraph{Retained quality under the clipped quantizer.}
Table~\ref{tab:suite} gives floating-point nDCG@10 for every encoder and corpus.
It also gives the fraction of it retained at $b{=}4$ and $b{=}3$, and top-1 agreement with the floating-point ranking at the same bit-widths.
Table~\ref{tab:quantizer} gives the same quantities under max-abs scaling next to the clipped values.

\begin{table}[t]
\centering
\caption{Retrieval quality retained of our two-round protocol under clipped quantization, as the percentage of floating-point nDCG@10 capped at $100$, for $b \in \{2,3,4,8\}$ on SciFact and NFCorpus, ranked per chunk and scored per document.}
\label{tab:quality-sweep}
\footnotesize
\setlength{\tabcolsep}{3pt}
\begin{tabular}{@{}c cc cccc c cc cccc@{}}
\toprule
& \multicolumn{6}{c}{SciFact} & & \multicolumn{6}{c}{NFCorpus} \\
\cmidrule(lr){2-7}\cmidrule(lr){9-14}
& \multicolumn{2}{c}{low-bit} & \multicolumn{4}{c}{conventional} & & \multicolumn{2}{c}{low-bit} & \multicolumn{4}{c}{conventional} \\
\cmidrule(lr){2-3}\cmidrule(lr){4-7}\cmidrule(lr){9-10}\cmidrule(lr){11-14}
$b$ & arctic & Cohere & BGE & mxbai & e5 & gte & & arctic & Cohere & BGE & mxbai & e5 & gte \\
\midrule
$2$ & $92.2$ & $89.7$ & $86.4$ & $88.5$ & $82.2$ & $86.6$ &  & $88.2$ & $88.6$ & $85.9$ & $86.5$ & $78.8$ & $83.6$ \\
$3$ & $98.7$ & $99.1$ & $97.4$ & $98.7$ & $97.3$ & $97.5$ &  & $98.2$ & $99.3$ & $99.0$ & $97.4$ & $94.1$ & $97.6$ \\
$\mathbf{4}$ & $99.9$ & $99.4$ & $99.7$ & $99.4$ & $100.0$ & $99.8$ &  & $99.6$ & $99.8$ & $99.5$ & $98.9$ & $96.5$ & $100.0$ \\
$8$ & $99.9$ & $99.5$ & $99.9$ & $100.0$ & $100.0$ & $99.8$ &  & $99.7$ & $100.0$ & $99.7$ & $100.0$ & $98.7$ & $100.0$ \\
\bottomrule
\end{tabular}
\end{table}

\begin{table}[ht]
\centering
\caption{The encoder suite under the clipped quantizer. Floating-point nDCG@10, the percentage of it retained at $b{=}4$ and $b{=}3$, and top-1 agreement with the floating-point ranking at the same bit-widths. \cmark: trained for low-bit readout.}
\label{tab:suite}
\footnotesize
\setlength{\tabcolsep}{3pt}
\begin{tabular}{@{}lc ccccc ccccc@{}}
\toprule
& & \multicolumn{5}{c}{SciFact} & \multicolumn{5}{c}{NFCorpus} \\
\cmidrule(lr){3-7}\cmidrule(lr){8-12}
Encoder & low-bit & float & ret.\ $4$ & ret.\ $3$ & top-1 $4$ & top-1 $3$ & float & ret.\ $4$ & ret.\ $3$ & top-1 $4$ & top-1 $3$ \\
\midrule
\texttt{arctic-embed-l-v2.0} & \cmark & $0.706$ & $100.4$ & $99.4$ & $0.96$ & $0.93$ & $0.354$ & $99.6$ & $98.0$ & $0.91$ & $0.82$ \\
\texttt{Cohere embed-v3} & \cmark & $0.718$ & $99.6$ & $98.7$ & $0.96$ & $0.91$ & $0.387$ & $99.5$ & $98.9$ & $0.86$ & $0.78$ \\
\texttt{bge-large-en-v1.5} & \xmark & $0.746$ & $99.4$ & $97.1$ & $0.94$ & $0.87$ & $0.381$ & $99.5$ & $98.6$ & $0.83$ & $0.71$ \\
\texttt{mxbai-embed-large-v1} & \xmark & $0.739$ & $100.3$ & $100.1$ & $0.93$ & $0.89$ & $0.387$ & $100.1$ & $98.9$ & $0.84$ & $0.73$ \\
\texttt{e5-large-v2} & \xmark & $0.722$ & $99.4$ & $96.8$ & $0.88$ & $0.77$ & $0.372$ & $96.3$ & $90.8$ & $0.79$ & $0.64$ \\
\texttt{gte-large} & \xmark & $0.743$ & $99.7$ & $97.8$ & $0.94$ & $0.81$ & $0.383$ & $100.0$ & $95.2$ & $0.81$ & $0.62$ \\
\bottomrule
\end{tabular}
\end{table}

\begin{table}[ht]
\centering
\caption{The quantizer lever. Each cell reads max-abs scaling $\to$ clipping at the $99$th percentile. We give the percentage of floating-point nDCG@10 retained and top-1 agreement at $b{=}4$ and $b{=}3$. Nothing on the cryptographic side differs between the two columns of each pair.}
\label{tab:quantizer}
\footnotesize
\setlength{\tabcolsep}{5pt}
\begin{tabular}{@{}l cc cc@{}}
\toprule
& \multicolumn{2}{c}{$b{=}4$} & \multicolumn{2}{c}{$b{=}3$} \\
\cmidrule(lr){2-3}\cmidrule(lr){4-5}
Encoder & ret.\ (\%) & top-1 & ret.\ (\%) & top-1 \\
\midrule
\multicolumn{5}{@{}l}{\emph{SciFact}}\\
\texttt{arctic-embed-l-v2.0} & $100.0 \to 100.4$ & $0.93 \to 0.96$ & $97.6 \to 99.4$ & $0.83 \to 0.93$ \\
\texttt{Cohere embed-v3} & $100.0 \to 99.6$ & $0.93 \to 0.96$ & $95.4 \to 98.7$ & $0.74 \to 0.91$ \\
\texttt{bge-large-en-v1.5} & $97.1 \to 99.4$ & $0.79 \to 0.94$ & $77.3 \to 97.1$ & $0.54 \to 0.87$ \\
\texttt{mxbai-embed-large-v1} & $99.5 \to 100.3$ & $0.81 \to 0.93$ & $82.3 \to 100.1$ & $0.57 \to 0.89$ \\
\texttt{e5-large-v2} & $96.3 \to 99.4$ & $0.79 \to 0.88$ & $82.4 \to 96.8$ & $0.56 \to 0.77$ \\
\texttt{gte-large} & $97.3 \to 99.7$ & $0.81 \to 0.94$ & $88.1 \to 97.8$ & $0.68 \to 0.81$ \\
\midrule
\multicolumn{5}{@{}l}{\emph{NFCorpus}}\\
\texttt{arctic-embed-l-v2.0} & $98.5 \to 99.6$ & $0.83 \to 0.91$ & $96.2 \to 98.0$ & $0.69 \to 0.82$ \\
\texttt{Cohere embed-v3} & $98.5 \to 99.5$ & $0.83 \to 0.86$ & $96.4 \to 98.9$ & $0.61 \to 0.78$ \\
\texttt{bge-large-en-v1.5} & $96.0 \to 99.5$ & $0.67 \to 0.83$ & $73.4 \to 98.6$ & $0.35 \to 0.71$ \\
\texttt{mxbai-embed-large-v1} & $96.2 \to 100.1$ & $0.63 \to 0.84$ & $77.4 \to 98.9$ & $0.37 \to 0.73$ \\
\texttt{e5-large-v2} & $93.0 \to 96.3$ & $0.65 \to 0.79$ & $78.4 \to 90.8$ & $0.37 \to 0.64$ \\
\texttt{gte-large} & $96.3 \to 100.0$ & $0.62 \to 0.81$ & $85.0 \to 95.2$ & $0.42 \to 0.62$ \\
\bottomrule
\end{tabular}
\end{table}

\paragraph{Why some encoders lose more than others.}
Quantization adds a small rounding error to every score.
A query keeps its top-1 document if the gap between its two best scores is larger than this error.
Rounding to a step $s$ gives an error that is uniform on $[-s/2, s/2]$, whose variance is $s^2/12$.
The error on a score therefore has standard deviation $\sqrt{(s_{\mat{A}}^2 + s_{\vect{x}}^2)/12}$, where $s_{\mat{A}}$ is the corpus step and $s_{\vect{x}}$ the query step.
For each query, we divide the gap by this standard deviation. We then take the median over queries.
Table~\ref{tab:margin} gives this ratio for every encoder at $b{=}4$.
The larger the ratio, the more often the top-1 document survives.
Over the $12$ encoder-corpus pairs, the ratio and the measured top-1 agreement have Spearman correlation $0.94$ at $b{=}4$ and $0.99$ at $b{=}3$.
The two encoders trained for low-bit readout have the largest ratios. BGE and mxbai come next, and e5 and gte have the smallest.
The table also explains why clipping helps BGE and mxbai.
Their largest coordinate holds over $5\%$ of the squared norm, against at most $1.9\%$ for the other encoders.
Under max-abs scaling, this one coordinate sets the step.
For BGE on SciFact, the step would be $0.305/7 \approx 0.044$ instead of the clipped $0.011$, i.e. about $4\times$ coarser.

\begin{table}[ht]
\centering
\caption{Score gap against rounding error at $b{=}4$, clipped quantizer. $\max|\mat{A}|$: largest corpus coordinate. $s_{\mat{A}}$: clipped step. top dim.: share of the squared norm in the largest coordinate. gap: median gap between the two best floating-point scores. gap/noise: median ratio of that gap to the rounding error. top-1: top-1 agreement with floating point.}
\label{tab:margin}
\small
\begin{tabular}{@{}llcccccc@{}}
\toprule
Encoder & Corpus & $\max|\mat{A}|$ & $s_{\mat{A}}$ & top dim.\ (\%) & gap & gap/noise & top-1 \\
\midrule
\texttt{arctic-embed-l-v2.0} & SciFact & $0.200$ & $0.0122$ & $0.5$ & $0.0582$ & $11.72$ & $0.96$ \\
\texttt{Cohere embed-v3} & SciFact & $0.221$ & $0.0125$ & $0.8$ & $0.0445$ & $8.83$ & $0.96$ \\
\texttt{bge-large-en-v1.5} & SciFact & $0.305$ & $0.0110$ & $5.5$ & $0.0319$ & $7.00$ & $0.94$ \\
\texttt{mxbai-embed-large-v1} & SciFact & $0.311$ & $0.0111$ & $5.4$ & $0.0345$ & $7.73$ & $0.93$ \\
\texttt{e5-large-v2} & SciFact & $0.170$ & $0.0105$ & $1.9$ & $0.0149$ & $3.48$ & $0.88$ \\
\texttt{gte-large} & SciFact & $0.172$ & $0.0107$ & $1.7$ & $0.0189$ & $4.38$ & $0.94$ \\
\midrule
\texttt{arctic-embed-l-v2.0} & NFCorpus & $0.194$ & $0.0123$ & $0.5$ & $0.0268$ & $5.22$ & $0.91$ \\
\texttt{Cohere embed-v3} & NFCorpus & $0.216$ & $0.0126$ & $1.0$ & $0.0211$ & $4.19$ & $0.86$ \\
\texttt{bge-large-en-v1.5} & NFCorpus & $0.287$ & $0.0109$ & $5.7$ & $0.0117$ & $2.64$ & $0.83$ \\
\texttt{mxbai-embed-large-v1} & NFCorpus & $0.292$ & $0.0110$ & $5.8$ & $0.0135$ & $2.97$ & $0.84$ \\
\texttt{e5-large-v2} & NFCorpus & $0.165$ & $0.0105$ & $1.9$ & $0.0063$ & $1.50$ & $0.79$ \\
\texttt{gte-large} & NFCorpus & $0.175$ & $0.0107$ & $1.7$ & $0.0073$ & $1.67$ & $0.81$ \\
\bottomrule
\end{tabular}
\end{table}

\begin{table}[ht]
\centering
\caption{Table~\ref{tab:quality-sweep} on whole abstracts instead of chunks. This checks comparability against published document-level baselines. We give the percentage of floating-point nDCG@10 retained under clipped quantization for each bit-width $b$.}
\label{tab:quality-sweep-abstracts}
\small
\begin{tabular}{@{}c cc cccc@{}}
\toprule
& \multicolumn{2}{c}{Trained for low-bit readout} & \multicolumn{4}{c}{Conventional} \\
\cmidrule(lr){2-3}\cmidrule(lr){4-7}
$b$ & arctic & Cohere & BGE & mxbai & e5 & gte \\
\midrule
\multicolumn{7}{@{}l}{\emph{SciFact}}\\
$2$ & $91.8$ & $92.1$ & $87.0$ & $90.4$ & $84.0$ & $87.6$ \\
$3$ & $99.4$ & $98.7$ & $97.1$ & $100.1$ & $96.8$ & $97.8$ \\
$\mathbf{4}$ & $100.4$ & $99.6$ & $99.4$ & $100.3$ & $99.4$ & $99.7$ \\
$5$ & $100.7$ & $99.8$ & $99.0$ & $99.9$ & $101.3$ & $99.2$ \\
$6$ & $100.1$ & $99.7$ & $99.4$ & $100.0$ & $101.2$ & $99.7$ \\
$7$ & $100.2$ & $99.8$ & $99.2$ & $100.0$ & $100.9$ & $99.7$ \\
$8$ & $100.0$ & $99.7$ & $99.1$ & $100.2$ & $100.9$ & $99.8$ \\
\midrule
\multicolumn{7}{@{}l}{\emph{NFCorpus}}\\
$2$ & $87.2$ & $88.7$ & $83.2$ & $87.4$ & $75.0$ & $82.7$ \\
$3$ & $98.0$ & $98.9$ & $98.6$ & $98.9$ & $90.8$ & $95.2$ \\
$\mathbf{4}$ & $99.6$ & $99.5$ & $99.5$ & $100.1$ & $96.3$ & $100.0$ \\
$5$ & $99.2$ & $99.8$ & $100.6$ & $100.7$ & $97.8$ & $98.7$ \\
$6$ & $99.9$ & $99.7$ & $100.4$ & $101.0$ & $98.5$ & $99.8$ \\
$7$ & $99.9$ & $99.8$ & $100.5$ & $100.7$ & $98.6$ & $99.7$ \\
$8$ & $99.6$ & $99.7$ & $100.4$ & $100.8$ & $98.5$ & $99.7$ \\
\bottomrule
\end{tabular}
\end{table}

\paragraph{Chunk-level retrieval.}
We first judge the chunked SciFact corpus per document, as in Table~\ref{tab:quality-sweep}.
Top-1 agreement at $b{=}4$ and $b{=}3$ is then $0.94$ and $0.87$ for BGE and $0.96$ and $0.92$ for arctic.
These are the same values as on whole abstracts (Table~\ref{tab:quality-sweep-abstracts}).
We then judge it per chunk, counting every chunk of a relevant abstract as relevant.
Top-1 agreement is then $0.92$ and $0.82$ for BGE and $0.93$ and $0.89$ for arctic.
The difference comes from chunks of one abstract trading places under rounding.
The reader of \S\ref{subsec:downstream} sees this reordering, but a document ranking does not.

\paragraph{Quality at scale.}
Tables~\ref{tab:scale-full} and~\ref{tab:scale-recall5} give the measurements behind the quality-at-scale paragraph of \S\ref{subsec:retrieval}.
They cover the five encoders of the suite that run locally, namely arctic, BGE, mxbai, e5 and gte.
Cohere is hosted and was not run at this scale.
The corpus is the BEIR Natural Questions passage collection ($2{,}681{,}468$ passages) with its $3{,}452$ test queries.
We draw subsets of $10^4$, $10^5$ and $10^6$ passages at random with a fixed seed.
Each subset contains every passage relevant to any query. As such, nDCG is defined at every size.
Embeddings are computed once per encoder on one H100 and stored in half precision.
The clipped percentile is estimated from a random sample of $4{,}096$ rows of each subset.
Quantized scores are computed in single precision.
This is exact for $b \leq 8$, since every partial sum is an integer below $2^{24}$.
We also check the scores against 64-bit integer arithmetic.
Ties are broken toward the lower index.

\paragraph{Recall at the reader's $k$.}
The sweep above scores at $k{=}10$.
The readers of \S\ref{subsec:downstream} consume five passages.
As such, we repeated the sweep at $k{=}5$ and added $b{=}2$.
Table~\ref{tab:scale-recall5} gives Recall@5 retained at every size.
At $b{=}4$, four of the five encoders retain at least $98.8\%$ of floating-point Recall@5 at every size.
The fifth, e5, retains $97.3\%$ at $2.68$M.
At $b{=}3$, arctic, BGE and mxbai retain at least $98.4\%$.
At $2.68$M, gte and e5 fall to $94.5$ and $92.0\%$.
At $b{=}2$, the loss grows with the corpus.
It goes from $1$ to $4\%$ at $10^4$ to $10$ to $23\%$ at $2.68$M.
The lead of arctic over each conventional encoder also grows, from $0.8$ to $2.7$ points at $10^4$ to $2.1$ to $12.8$ points at $2.68$M.
This is the retrieval-side counterpart of the downstream drop of \S\ref{subsec:downstream}.
On the $300$ NQ questions of that section, the top-$5$ set of BGE contained a relevant passage for $70\%$ of questions in floating point and $60\%$ at $b{=}2$.
This is $86\%$ retained, against $86.1\%$ measured here on all $3{,}452$ queries.

\begin{table}[ht]
\centering
\caption{Recall@5 retained under the clipped quantizer versus corpus size on BEIR Natural Questions ($3{,}452$ queries), at $b \in \{2,3,4\}$. Floating-point Recall@5 is given for reference. Generated by \texttt{implem/analysis/scale\_recall5\_table.py}.}
\label{tab:scale-recall5}
\footnotesize
\setlength{\tabcolsep}{5pt}
\begin{tabular}{@{}ll c ccc@{}}
\toprule
& & & \multicolumn{3}{c}{Recall@5 retained (\%)} \\
\cmidrule(lr){4-6}
Encoder & $m$ & Recall@5 (float) & $b{=}2$ & $b{=}3$ & $b{=}4$ \\
\midrule
\texttt{arctic-embed-l-v2.0} & $10^4$ & $0.988$ & $99.2$ & $99.9$ & $100.0$ \\
 & $10^5$ & $0.967$ & $97.0$ & $99.9$ & $99.9$ \\
 & $10^6$ & $0.848$ & $93.1$ & $99.5$ & $99.7$ \\
 & $2.68{\cdot}10^6$ & $0.737$ & $89.9$ & $99.0$ & $99.8$ \\
\midrule
\texttt{bge-large-en-v1.5} & $10^4$ & $0.982$ & $97.7$ & $99.9$ & $100.1$ \\
 & $10^5$ & $0.944$ & $94.7$ & $99.2$ & $99.5$ \\
 & $10^6$ & $0.794$ & $88.3$ & $98.7$ & $98.8$ \\
 & $2.68{\cdot}10^6$ & $0.655$ & $86.1$ & $98.5$ & $99.7$ \\
\midrule
\texttt{mxbai-embed-large-v1} & $10^4$ & $0.983$ & $98.4$ & $99.9$ & $100.1$ \\
 & $10^5$ & $0.944$ & $96.3$ & $99.6$ & $99.9$ \\
 & $10^6$ & $0.800$ & $90.4$ & $98.8$ & $99.4$ \\
 & $2.68{\cdot}10^6$ & $0.661$ & $87.7$ & $98.4$ & $100.6$ \\
\midrule
\texttt{e5-large-v2} & $10^4$ & $0.988$ & $96.5$ & $99.5$ & $99.8$ \\
 & $10^5$ & $0.961$ & $92.1$ & $98.5$ & $99.5$ \\
 & $10^6$ & $0.847$ & $82.7$ & $95.8$ & $98.3$ \\
 & $2.68{\cdot}10^6$ & $0.745$ & $77.1$ & $92.0$ & $97.3$ \\
\midrule
\texttt{gte-large} & $10^4$ & $0.986$ & $97.2$ & $99.7$ & $99.8$ \\
 & $10^5$ & $0.945$ & $94.4$ & $99.4$ & $100.0$ \\
 & $10^6$ & $0.795$ & $85.9$ & $96.3$ & $99.9$ \\
 & $2.68{\cdot}10^6$ & $0.658$ & $82.6$ & $94.5$ & $99.6$ \\
\bottomrule
\end{tabular}
\end{table}

\begin{table}[ht]
\centering
\caption{Quality versus corpus size on BEIR Natural Questions with the clipped quantizer. We give nDCG@10 retained at $b{=}3$ and $b{=}4$. We also give the smallest $b$ that retains $99\%$ of floating-point nDCG@10, or $>8$ when no $b$ up to $8$ does. Generated by \texttt{implem/analysis/scale\_analysis.py}.}
\label{tab:scale-full}
\footnotesize
\setlength{\tabcolsep}{5pt}
\begin{tabular}{@{}ll cc c@{}}
\toprule
& & \multicolumn{2}{c}{nDCG@10 ret.\ (\%)} & required $b$ \\
\cmidrule(lr){3-4}
Encoder & $m$ & $b{=}3$ & $b{=}4$ & (top-$k$) \\
\midrule
\texttt{arctic-embed-l-v2.0} & $10^4$ & $100.0$ & $99.9$ & $3$ \\
 & $10^5$ & $99.9$ & $100.1$ & $3$ \\
 & $10^6$ & $99.1$ & $99.9$ & $3$ \\
 & $2.68{\cdot}10^6$ & $98.5$ & $99.7$ & $4$ \\
\midrule
\texttt{bge-large-en-v1.5} & $10^4$ & $99.5$ & $99.9$ & $3$ \\
 & $10^5$ & $99.0$ & $99.6$ & $3$ \\
 & $10^6$ & $98.0$ & $99.1$ & $4$ \\
 & $2.68{\cdot}10^6$ & $98.1$ & $99.3$ & $4$ \\
\midrule
\texttt{mxbai-embed-large-v1} & $10^4$ & $99.8$ & $100.0$ & $3$ \\
 & $10^5$ & $99.5$ & $100.0$ & $3$ \\
 & $10^6$ & $98.9$ & $100.1$ & $4$ \\
 & $2.68{\cdot}10^6$ & $98.3$ & $99.9$ & $4$ \\
\midrule
\texttt{e5-large-v2} & $10^4$ & $99.0$ & $99.7$ & $3$ \\
 & $10^5$ & $97.7$ & $99.3$ & $4$ \\
 & $10^6$ & $93.9$ & $98.1$ & $8$ \\
 & $2.68{\cdot}10^6$ & $91.9$ & $97.4$ & $>8$ \\
\midrule
\texttt{gte-large} & $10^4$ & $99.4$ & $99.9$ & $3$ \\
 & $10^5$ & $98.4$ & $99.9$ & $4$ \\
 & $10^6$ & $95.7$ & $99.9$ & $4$ \\
 & $2.68{\cdot}10^6$ & $94.0$ & $100.6$ & $4$ \\
\bottomrule
\end{tabular}
\end{table}

\end{document}